\documentclass[11pt, a4paper]{article}
\usepackage[english]{babel}
\usepackage[pagewise]{lineno}
\usepackage{amsmath}
\usepackage{amsfonts}
\usepackage{graphicx}
\usepackage{enumitem} 
\usepackage[utf8]{inputenc}
\usepackage{url}
\usepackage{amsthm,amssymb}
\usepackage[top=1in, bottom=1in, left=1in, right=1in]{geometry}
\numberwithin{figure}{section}
\newtheorem{theorem}{Theorem}[section]
\newtheorem{lemma}[theorem]{Lemma}
\newtheorem{corollary}[theorem]{Corollary}
\newtheorem{proposition}[theorem]{Proposition}
\newtheorem{example}[theorem]{Example}
\newtheorem{remark}[theorem]{Remark}
\newtheorem{definition}[theorem]{Definition}

\begin{document}

\title{On $\mathbb{Z}_{p^r}\mathbb{Z}_{p^r}\mathbb{Z}_{p^s}$-Additive Cyclic Codes\footnote{Email: cristina.fernandez@uab.es (C. Fernández-Córdoba), sachinp@iiitd.ac.in (S. Pathak), upadhyay@bhu.ac.in (A.K. Upadhyay)}}
\author{ Cristina Fernández-Córdoba$^1$, Sachin Pathak$^2$\footnote{Corresponding author.},  Ashish Kumar Upadhyay$^3$}
\date{
\small{
1. Department of Information and Communications Engineering, Universitat Aut\'{o}noma de Barcelona, 08193-Bellaterra, Spain.\\
2. Department of Mathematics, IIIT-Delhi, New Delhi 110020, India.\\
3. Department of Mathematics, Institute of Science, Banaras Hindu University, Varanasi 221005,
India.
}
\today}
\maketitle

\begin{abstract}
In this paper, we introduce $\mathbb{Z}_{p^r}\mathbb{Z}_{p^r}\mathbb{Z}_{p^s}$-additive cyclic codes for $r\leq s$. These codes can be identified as $\mathbb{Z}_{p^s}[x]$-submodules of $\mathbb{Z}_{p^r}[x]/\langle x^{\alpha}-1\rangle \times \mathbb{Z}_{p^r}[x]/\langle x^{\beta}-1\rangle\times \mathbb{Z}_{p^s}[x]/\langle x^{\gamma}-1\rangle$. We determine the generator polynomials and minimal generating sets for this family of codes. Some previous works has been done for the case $p=2$ with $r=s=1$, $r=s=2$, and $r=1,s=2$. However, we show that in these previous works the classification of these codes were incomplete and the statements in this paper complete such classification. We also discuss the structure of separable $\mathbb{Z}_{p^r}\mathbb{Z}_{p^r}\mathbb{Z}_{p^s}$-additive cyclic codes and determine their generator polynomials. Further, we also study the duality of $\mathbb{Z}_{p^s}[x]$-submodules. As applications, we present some examples and construct some optimal binary codes. 
     
\end{abstract}


\section{Introduction}
 Codes over finite rings have been an interesting topic of research for many researchers in the last few years, especially after the work of Hammons et al. \cite{hammons}. Cyclic codes and one of their significant generalization, constacyclic codes, are important classes of codes because of their rich algebraic structure and because of their ease of implementation. The structure of constacyclic codes has been studied by many researchers in \cite{Siap2014,AS09, abu2007,bonnecaze, dinh,QZZ06,zhu} over several finite rings.
 In all of these works, the codes have been studied over single alphabets.\par
 In 1997, Rifà et al. \cite{RP97} first described codes over the mixed alphabets. After that, in 1998, Brouwer et al. \cite{bro1998} considered the mixed alphabets  $\mathbb{Z}_2$ and $\mathbb{Z}_3$, and studied the lower and upper bound for the maximal possible size of error-correcting codes over mixed alphabets. In 1973, Delsarte \cite{D73} introduced additive codes in temrs of association schemes. In general, an additive code is defined as a subgroup of the underlying abelian group. A few years ago, Borges et al. \cite{Bor2009} studied $\mathbb{Z}_2\mathbb{Z}_4$-additive codes. For these codes, they partitioned the set of coordinates into two parts, the first of which corresponds to binary coordinates and the second one to quaternary coordinates. They defined $\mathbb{Z}_2 \mathbb{Z}_4$-additive codes as the subgroup of the group $\mathbb{Z}^{\alpha}_2\times \mathbb{Z}^{\beta}_4$. These codes are a generalization of binary and quaternary codes for $\beta=0$ and $\alpha=0$, respectively. In 2014, T. Abualrab et al. \cite{Th14}, studied $\mathbb{Z}_2\mathbb{Z}_4$-additive cyclic codes. They obtained the minimal generating set for these codes and some optimal bound from this family of codes. After that, in 2016, Borges et al. \cite{vall2106} studied the generating polynomials for dual codes of $\mathbb{Z}_2\mathbb{Z}_4$-additive cyclic codes. Very recently, Borges et al. \cite{Borg2018} studied binary images of additive cyclic codes over the product of chain rings.  Some generalizations of $\mathbb{Z}_2\mathbb{Z}_4$-additive codes and other related codes can be seen in \cite{Abu2015, Ay2016,Ay2017,siap2015}.\par 
  In 2018, Borges et al. \cite{BCV18} studied the generalization of $\mathbb{Z}_2\mathbb{Z}_4$-additive cyclic codes over $\mathbb{Z}_{p^r}\mathbb{Z}_{p^s}$-additive cyclic codes and defined the structure of these codes. They also discussed minimal generating sets for these codes and some properties of their duals.  In 2018, Borges et al. \cite{bor2014} introduced $\mathbb{Z}_2$-double cyclic codes and determined the generator polynomials for these codes. Further, they determined the minimal generating sets of these codes and generator polynomials for their duals and obtained several binary optimal codes from their study. They established the relationship between the generator polynomials of $\mathbb{Z}_2$-double cyclic codes and their duals. Moreover, they studied that $\mathbb{Z}_2$-double cyclic codes are related to $\mathbb{Z}_4$-cyclic codes and $\mathbb{Z}_2\mathbb{Z}_4$-additve cyclic codes. On the line of double cyclic codes, Gao et al. \cite{shi2105} and  Yao et al. \cite{yao2015} introduced the structure of double cyclic codes over $\mathbb{Z}_4$ and $\mathbb{F}_q+u\mathbb{F}_q+u^2\mathbb{F}_q$, respectively. They determined generator polynomials for these codes as well as their duals. They also studied the minimal generating sets for these codes. As a generalization, Mostafanasab \cite{mosta} studied the triple cyclic codes over $\mathbb{Z}_2$ and determined structural properties of triple cyclic codes over $\mathbb{Z}_2$. Wu et al. \cite{Gao2016} extended the work of Mostafanasab \cite{mosta} and studied triple cyclic codes over $\mathbb{Z}_4$. In 2018, Wu et al. \cite{224} studied $\mathbb{Z}_2\mathbb{Z}_2\mathbb{Z}_4$-additive cyclic codes, and the generator polynomials of these codes along with their duals. Inspired by \cite{BCV18, mosta, Gao2016, 224}, in this paper, we study $\mathbb{Z}_{p^r}\mathbb{Z}_{p^r}\mathbb{Z}_{p^s}$-additive cyclic codes, where $\mathbb{Z}_{p^n}$ is a ring of integers modulo $p^n$. Both the rings $\mathbb{Z}_{p^r}$ and $\mathbb{Z}_{p^s}$ are finite chain rings with maximal ideal $\langle p\rangle$. We study the algebraic structure of $\mathbb{Z}_{p^r}\mathbb{Z}_{p^r}\mathbb{Z}_{p^s}$-additive cyclic codes in details. In \cite{mosta, Gao2016, 224}, authors studied additive cyclic codes over mixed alphabets considering the product of three not necessarily different rings. They defined the codes over the product of rings and gave their generator polynomials. However, as we show in this paper, these generator polynomials do not define the general case and we give the general form of generator polynomials.\par
As applications, we provide some optimal binary codes in Table \ref{t1}. This paper is organized as follows: In Section 2, we introduce cyclic codes over $\mathbb{Z}_{p^r}$ and determine generating sets of a cyclic code over $\mathbb{Z}_{p^r}$. In Section 3, we discuss the structure of $\mathbb{Z}_{p^r}\mathbb{Z}_{p^r}\mathbb{Z}_{p^s}$-additive codes. In Section 4, we study the algebraic structure of $\mathbb{Z}_{p^r}\mathbb{Z}_{p^r}\mathbb{Z}_{p^s}$-additive cyclic codes and separable codes and discuss their generator polynomials. Further, we determine the generator polynomials of each class of additive codes separately. Moreover, in Example \ref{ex:Z2Z2Z4} and Remark \ref{R4.12}, we show that the code presented in Example \ref{ex:Z2Z2Z4} can not be generated by the generator polynomials discussed in previously published work. In Section 5, we discuss the minimal generating sets for this family of codes and present an example. Section 6 is devoted to the study of the duality of $\mathbb{Z}_{p^r}\mathbb{Z}_{p^r}\mathbb{Z}_{p^s}$-additive cyclic codes and we determine the relation between $\mathbb{Z}_{p^r}\mathbb{Z}_{p^r}\mathbb{Z}_{p^s}$-additive cyclic codes and dual of their duals. Section 7 concludes this paper.

\section{Cyclic Codes Over $\mathbb{Z}_{p^a}$}
Let $\mathbb{Z}_{p^a}$ be the finite ring of integers modulo $p^a$, where $p$ is a prime number. A non-empty subset of $\mathbb{Z}^n_{p^a}$ is said to be a code, and a submodule of $\mathbb{Z}^n_{p^a}$ is said to be linear code of length $n$ over $\mathbb{Z}_{p^a}$. The elements of linear code are called codewords. Let $C$ be a  linear code of length $n$ over $\mathbb{Z}_{p^a}$. Then $C$ is said to be a cyclic code of length $n$ if it has the property that for any codeword $(c_0,c_1,\cdots,c_{n-1})\in C$ the cyclic shift $(c_{n-1},c_0,\cdots,c_{n-2})\in C$.\par
Let $f_0, f_1, \cdots, f_{s-1}$ be polynomials in a $\mathbb{Z}_{p^a}[x]$-module. We denote by $\langle f_0, f_1, \cdots, f_{s-1}\rangle$ the $\mathbb{Z}_{p^a}[x]$-submodule and $\langle f_0, f_1, \cdots, f_{s-1}\rangle_{p^a}$ the $\mathbb{Z}_{p^a}$-submodule, respectively generated by  $f_0, f_1, \cdots, f_{s-1}$.\par
Let $C$ be a cyclic code of length $n$ over $\mathbb{Z}_{p^a}$. Every codeword $\mathbf c= (c_0,c_1,\cdots,c_{n-1})\in C$ can be naturally associated with the polynomial $\mathbf{c}(x)= c_0+c_1x+\cdots+c_{n-1}x^{n-1}$ in $\frac{\mathbb{Z}_{p^a}[x]}{\langle x^n-1\rangle}$. If we multiply by $x\in \mathbb{Z}_{p^a}[x]$ any element $\mathbf c(x)\in \mathbb{Z}_{p^a}[x]/\langle x^n-1\rangle$ such that $x\mathbf{c}(x)=c_{n-1}+c_0x+\cdots+c_{n-2}x^{n-1}$, then it is equivalent to the right shift of $\mathbf c(x)$ on $\mathbb{Z}^{n}_{p^a}$. So we can identify cyclic code $C$ as an ideal of $\frac{\mathbb{Z}_{p^a}[x]}{\langle x^n-1\rangle}$. We assume that $\mbox{gcd}(n,p)=1.$ Therefore, the polynomial $x^n-1$ can decompose as a product of pairwise coprime polynomials uniquely over $\mathbb{Z}_{p^a}[x]$.
 
\begin{theorem}{\cite[Corollary 3.5]{kanwar}}
Suppose $\mbox{gcd}(n,p)=1$ and $C$ is a cyclic code of length $n$ over $\mathbb{Z}_{p^a}$. Then there exist polynomials $f_0, f_1, \cdots, f_{a-1}$ such that $C= \langle f_0, pf_1,p^2f_2, \cdots, p^{a-1}f_{a-1}\rangle$ and $f_{a-1}|f_{a-2}|\cdots|f_1|f_0|(x^n-1)$.
\end{theorem}

Let $C=\langle f_0, pf_1,p^2f_2 \cdots, p^{a-1}f_{a-1}\rangle$ be a cyclic code of length $n$ over $\mathbb{Z}_{p^a}$, and let $f=f_0+pf_1+p^2f_2+ \cdots+ p^{a-1}f_{a-1}$. By definition, $f_0$ is a factor of $x^n-1$ and for $i= 0,1,\cdots, a-1$ the polynomial $f_i$ is a factor of $f_{i-1}$. We denote $\hat{f_0}=\frac{x^n-1}{f_0}$ and $\hat{f_i}= \frac{f_{i-1}}{f_i}$ for $i= 1,2,\cdots, a-1$. We define, $F= \prod_{i=1}^{a-1}\hat{f_i}$. It is clear that $Ff= (\prod_{i=1}^{a-1}\hat{f_i})f=0$ over $\mathbb{Z}_p^a[x]/\langle x^n-1\rangle$.
\begin{corollary}{\cite[Theorem 3.4]{dinh}}\label{cor1}
Let $C$ be a cyclic code of length $n$ over $\mathbb{Z}_{p^a}$ such that $C= \langle f_0, pf_1,p^2f_2,\linebreak \cdots, p^{a-1}f_{a-1}\rangle$ and $f_{a-1}|f_{a-2}|\cdots|f_1|f_0|(x^n-1)$. Then
\begin{equation*}
 |C|=p^{\sum_{i=0}^{a-1}(a-i)\deg(\hat{f_i})}.   
\end{equation*}
\end{corollary}
\begin{lemma}\cite[Lemma 2.2]{BCV18}
Let $C$ be a cyclic code of length $n$ over $\mathbb{Z}_{p^a}$. Let $f_0,f_1,\cdots, f_{a-1}$ polynomials in $\mathbb{Z}_{p^a}[x]$ such that $C= \langle f_0, pf_1,p^2f_2, \cdots, p^{a-1}f_{a-1}\rangle$ and $f_{a-1}|f_{a-2}|\cdots|f_1|f_0|(x^n-1)$, and let $f=f_0+pf_1+p^2f_2+ \cdots+ p^{a-1}f_{a-1}$. Then,
\begin{enumerate}
    \item $p^{a-1}f= p^{a-1}f_{a-1}\frac{F}{\hat{f_0}},$
    \item $p^{a-1-i}(\prod_{k=0}^{i-1}\hat{f_k})f= p^{a-1}f_{a-1}\frac{F}{\hat{f_i}},$ for $i= 1,2,\cdots, a-1.$
\end{enumerate}
\end{lemma}

The following result is well known for cyclic codes of length $n$ over $\mathbb{Z}_{p^a}$, where $p$ is a prime not dividing $n$. In \cite{Calder1995,kanwar}, it has been proved that $\mathbb{Z}_{p^a}[x]/\langle x^n-1\rangle$ is a principal ideal ring.
\begin{theorem}\label{th2.4}\cite{Calder1995}
Let $C$ be a cyclic code of length $n$ over $\mathbb{Z}_{p^a}$.  Let $f_0,f_1,\cdots, f_{a-1}$ polynomials in $\mathbb{Z}_{p^a}[x]$ such that $C= \langle f_0, pf_1,p^2f_2, \cdots, p^{a-1}f_{a-1}\rangle$ and $f_{a-1}|f_{a-2}|\cdots|f_1|f_0|(x^n-1)$. Then $C=\langle f \rangle$, where $f= f_0+pf_1+p^2f_2+ \cdots+ p^{a-1}f_{a-1}$. 
\end{theorem}
It is a well known result that $\mathbb{Z}_{p^a}[x]/\langle x^n-1\rangle$ is a $\mathbb{Z}_{p^a}[x]$-module. Furthermore, we know that $\mathbb{Z}^{n}_{p^a}$ is a $\mathbb{Z}_{p^a}$-module. Now we  discuss the set of generator polynomials for a cyclic code $C$ of length $n$ as a $\mathbb{Z}_{p^a}$-module in terms of polynomials.

\begin{theorem}\label{th2.5} \cite[Theorem 2.5]{BCV18}
Let $C$ be a cyclic code of length $n$ over $\mathbb{Z}_{p^a}$ and $C=\langle f \rangle = \langle f_0+pf_1+p^2f_2+ \cdots+ p^{a-1}f_{a-1}\rangle $ with $f_{a-1}|f_{a-2}|\cdots|f_1|f_0|(x^n-1)$. We define the following sets 
\begin{align*}
    S_0 &= \{x^if\}_{i=0}^{\deg(\hat{f}_0)-1}= \{x^i(f_0+pf_1+p^2f_2+ \cdots+ p^{a-1}f_{a-1})\}_{i=0}^{\deg(\hat{f}_0)-1}\\
    S_1   &= \{x^i\hat{f}_0f\}_{i=0}^{\deg(\hat{f}_1)-1}= \{x^i(pf_1\hat{f}_0+p^2f_2\hat{f}_0+ \cdots+ p^{a-1}f_{a-1}\hat{f}_0)\}_{i=0}^{\deg(\hat{f}_1)-1}\\
    &\vdots\\
    S_k   &= \Big\{x^i(\prod_{t=0}^{k-1}\hat{f}_t)f\Big\}_{i=0}^{\deg(\hat{f}_k)-1}\\
    &\vdots\\
    S_{a-1} &= \Big\{x^i(\prod_{t=0}^{a-2}\hat{f}_t)f\Big\}_{i=0}^{\deg(\hat{f}_{a-1})-1}.\\
    \end{align*}
    Then,
    \begin{equation*}
        S=\bigcup_{k=0}^{a-1}S_k
    \end{equation*}
    forms a minimal generating set for the cyclic code $C$ as a $\mathbb{Z}_{p^a}$-module.
\end{theorem}

\section{$\mathbb{Z}_{p^r}\mathbb{Z}_{p^r}\mathbb{Z}_{p^s}$-Additive  Codes}
 Let $\mathbb{Z}_{p^r}$ and $\mathbb{Z}_{p^s}$ be two rings of integer modulo $p^r$ and $p^s$, respectively, with $r\leq s$ and $p$ prime. We can prove that both the rings $\mathbb{Z}_{p^r}$ and $\mathbb{Z}_{p^s}$ have the same residue field $\mathbb{Z}_p$. Any element $v$ of $\mathbb{Z}_{p^r}$  can be written as $v= v_0+pv_1+p^2v_2+\cdots+p^{r-1}v_{r-1}$ uniquely, and an element $u$ of $\mathbb{Z}_{p^s}$ can be written as $u= u_0+pu_1+p^2u_2+\cdots+p^{s-1}u_{s-1}$ uniquely, where $u_i,v_j\in \mathbb{Z}_p$. This expansion is called $p$-adic expansion \cite{Calder1995}.\par
 Then we can have a surjective ring homomorphism from $\mathbb{Z}_{p^r}$ to $\mathbb{Z}_{p^s}$ such that
 \begin{align*}
     \phi&: \mathbb{Z}_{p^s}\longrightarrow \mathbb{Z}_{p^r}\\
          & u\longmapsto u ~\mbox{mod}~ p^r.
    \end{align*}
 For $i\geq r$, we note that $\phi(p^i)=0.$ For $u\in \mathbb{Z}_{p^s}$ and $v\in \mathbb{Z}_{p^r}$, we define the multiplication denoted by $*$ as follows:  $u*v= \phi(u)v$. Then, $\mathbb{Z}_{p^r}$ is a $\mathbb{Z}_{p^s}$-module with the multiplication $*$ given by $\phi$. Since the ring $\mathbb{Z}_{p^r}$ is a commutative ring, then the multiplication $*$ has the commutative property. We can generalize the above defined multiplication $*$ over the ring $\mathbb{Z}^{\alpha}_{p^r}\times \mathbb{Z}^{\beta}_{p^r}\times \mathbb{Z}^{\gamma}_{p^s}$ as follows. Let $b$ be any element of $\mathbb{Z}_{p^s}$ and $\mathbf{w}= (u~|~v~|~w)=(u_0,u_1,\cdots,u_{\alpha-1}~|~v_0,v_1,\cdots,v_{\beta-1}~|~w_0,w_1,\cdots,w_{\gamma-1})\in \mathbb{Z}^{\alpha}_{p^r}\times \mathbb{Z}^{\beta}_{p^r}\times \mathbb{Z}^{\gamma}_{p^s}$. Then,
 \[b*\mathbf{w}=(\phi(b)u_0,\phi(b)u_1,\cdots,\phi(b)u_{\alpha-1}~|~\phi(b)v_0,\phi(b)v_1,\cdots,\phi(b)v_{\beta-1}~|~bw_0,bw_1,\cdots,bw_{\gamma-1}).\]
 With that external multiplication, the ring $\mathbb{Z}^{\alpha}_{p^r}\times \mathbb{Z}^{\beta}_{p^r}\times \mathbb{Z}^{\gamma}_{p^s}$ forms a $\mathbb{Z}_{p^s}$-module. Now, we define additive codes as submodules.
 \begin{definition}
 A non-empty subset $C$ of $\mathbb{Z}^{\alpha}_{p^r}\times \mathbb{Z}^{\beta}_{p^r}\times \mathbb{Z}^{\gamma}_{p^s}$ is said to be a $\mathbb{Z}_{p^r}\mathbb{Z}_{p^r}\mathbb{Z}_{p^s}$-additive code of block length $(\alpha,\beta,\gamma)$ if $C$ is a $\mathbb{Z}_{p^s}$-submodule of $\mathbb{Z}^{\alpha}_{p^r}\times \mathbb{Z}^{\beta}_{p^r}\times \mathbb{Z}^{\gamma}_{p^s}$.
 \end{definition}
 Let $C_{\alpha}$ be the canonical projection of a $\mathbb{Z}_{p^r}\mathbb{Z}_{p^r}\mathbb{Z}_{p^s}$-additive code $C$ on first $\alpha$ coordinates, $C_{\beta}$ on the next $\beta$ coordinates and $C_{\gamma}$ on the last $\gamma$ coordinates. Then, clearly $C_{\alpha}$ and $C_{\beta}$ are linear codes of length $\alpha$ and $\beta$ over $\mathbb{Z}_{p^r}$, respectively, and $C_{\gamma}$ is a linear code of length $\gamma$ over $\mathbb{Z}_{p^s}$. The code $C$ is said to be  separable if $C$ is the direct product of $C_{\alpha}$, $C_{\beta}$ and $C_{\gamma}$ i.e. $C=C_{\alpha}\times C_{\beta}\times C_{\gamma}$.\par
 Since we have assumed that $r\leq s$, we can consider the inclusion map
 \begin{align*}
     \epsilon&: \mathbb{Z}_{p^r} \hookrightarrow \mathbb{Z}_{p^s}\\
                & a\longmapsto a.
    \end{align*}
 Let $\mathbf{w}, \mathbf{w'}\in \mathbb{Z}^{\alpha}_{p^r}\times \mathbb{Z}^{\beta}_{p^r}\times \mathbb{Z}^{\gamma}_{p^s}$, then we define the inner product as follows
 \begin{equation*}
     \mathbf{w}\cdot\mathbf{w'}= p^{s-r}\sum_{i=0}^{\alpha-1}\epsilon(u_iu'_i)+p^{s-r}\sum_{i=0}^{\beta-1}\epsilon(v_iv'_i)+\sum_{i=0}^{\gamma-1}w_iw'_i\in \mathbb{Z}_{p^s},
 \end{equation*}
 and the dual of $\mathbb{Z}_{p^r}\mathbb{Z}_{p^r}\mathbb{Z}_{p^s}$-additive code $C$ is defined as follows
 \begin{equation*}
     C^{\perp}= \{\mathbf{w'}\in \mathbb{Z}_{p^r}\times \mathbb{Z}_{p^r}\times \mathbb{Z}_{p^s} |~ \mathbf{w}\cdot\mathbf{w'}=0, \forall \mathbf{w}\in C \}.
 \end{equation*}
 Let $C$ be a separable code in $\mathbb{Z}^{\alpha}_{p^r}\times \mathbb{Z}^{\beta}_{p^r}\times \mathbb{Z}^{\gamma}_{p^s}$, then $C^{\perp}$ is also a separable code and $C^{\perp}=C_{\alpha}^{\perp}\times C_{\beta}^{\perp}\times C_{\gamma}^{\perp}.$ $C$ is called self-orthogonal if $C^{\perp} \subseteq C$ and self-dual if $C^{\perp}= C$.
 \section{The Structure of $\mathbb{Z}_{p^r}\mathbb{Z}_{p^r}\mathbb{Z}_{p^s}$-Additive Cyclic Codes}
 This section is dedicated to the study of additive cyclic codes. In this section, we define $\mathbb{Z}_{p^r}\mathbb{Z}_{p^r}\mathbb{Z}_{p^s}$-additive cyclic codes and discuss structural properties of these codes. Further, we discuss the generator polynomials of this family of codes and the structure of separable codes has also been discussed.
 \begin{definition}
 Let $C\subseteq \mathbb{Z}^{\alpha}_{p^r}\times \mathbb{Z}^{\beta}_{p^r}\times \mathbb{Z}^{\gamma}_{p^s}$ be a $\mathbb{Z}_{p^r}\mathbb{Z}_{p^r}\mathbb{Z}_{p^s}$-additive codes of block length $(\alpha, \beta, \gamma)$. The code $C$ is said to be cyclic if for any
 \[(u_0,u_1,\cdots,u_{\alpha-1}~|~v_0,v_1,\cdots,v_{\beta-1}~|~w_0,w_1,\cdots,w_{\gamma-1})\in C,\]
 its cyclic shift
 \[(u_{\alpha-1},u_0,u_1,\cdots,u_{\alpha-2}~|~v_{\beta-1},v_0,v_1,\cdots,v_{\beta-2}~|~w_{\gamma-1},w_0,w_1,\cdots,w_{\gamma-2})\in C.\]
 \end{definition}
 Let $\mathbf{w}=(u_0,u_1,\cdots,u_{\alpha-1}~|~v_0,v_1,\cdots,v_{\beta-1}~|~w_0,w_1,\cdots,w_{\gamma-1})\in C$ and $i$ be a positive integer. We define the $i^{th}$ shift of $\mathbf{w}$ by $\mathbf{w}^{(i)}= (u_{0-i},u_{1-i},\cdots,u_{\alpha-1-i}~|~v_{0-i},v_{1-i},\cdots,v_{\beta-1-i}~|~w_{0-i},w_{1-i},\cdots,$ $w_{\gamma-1-i})$, where we read the subscripts modulo $\alpha$, $\beta$ and $\gamma$, respectively. Now if $C\subseteq \mathbb{Z}^{\alpha}_{p^r}\times \mathbb{Z}^{\beta}_{p^r}\times \mathbb{Z}^{\gamma}_{p^s}$ is a $\mathbb{Z}_{p^r}\mathbb{Z}_{p^r}\mathbb{Z}_{p^s}$-additive cyclic code, then $C_{\alpha}$, $C_{\beta}$ and $C_{\gamma}$ are also cyclic code over $\mathbb{Z}_{p^r}$, $\mathbb{Z}_{p^r}$ and $\mathbb{Z}_{p^s}$ of length $\alpha$, $\beta$ and $\gamma$, respectively.\par
 We note that the definition of $\mathbb{Z}_{p^r}\mathbb{Z}_{p^r}\mathbb{Z}_{p^s}$-additive cyclic codes is well defined for both $\mathbb{Z}_{p^r}$ and $\mathbb{Z}_{p^s}$ different rings, since the elements of first $\alpha+\beta$ coordinates and last $\gamma$ coordinates belongs to the rings $\mathbb{Z}_{p^r}$ and $\mathbb{Z}_{p^s}$, respectively. We present some remarks for particular cases.
 \begin{remark}\em 
 If $\gamma=0$ then $\mathbb{Z}_{p^r}\mathbb{Z}_{p^r}\mathbb{Z}_{p^s}$-additive cyclic codes are known as double cyclic codes over $\mathbb{Z}_{p^r}$ \cite{bor2014, shi2105}.
 \end{remark}
 
 \begin{remark}\em 
 In particular case, if we have $r=s$ then the $\mathbb{Z}_{p^r}\mathbb{Z}_{p^r}\mathbb{Z}_{p^s}$-additive cyclic codes are known in the literature \cite{mosta, Gao2016} as triple cyclic codes.
 \end{remark}
 \begin{remark}\em
 If $\alpha=0$ or ($\beta=0$) then $\mathbb{Z}_{p^r}\mathbb{Z}_{p^r}\mathbb{Z}_{p^s}$-additive cyclic codes are known as $\mathbb{Z}_{p^r}\mathbb{Z}_{p^s}$-additive cyclic codes studied in \cite{BCV18}.
 \end{remark}
 \par
 \noindent We denote the ring $\mathbb{Z}_{p^r}/\langle x^{\alpha}-1\rangle \times\mathbb{Z}_{p^r}/\langle x^{\beta}-1\rangle \times \mathbb{Z}_{p^s}/\langle x^{\gamma}-1\rangle $ by $R_{\alpha,\beta,\gamma}$. There is a bijective map between $\mathbb{Z}_{p^r}\times \mathbb{Z}_{p^r}\times \mathbb{Z}_{p^s}$ and $R_{\alpha,\beta,\gamma}$, where $\mathbf{w}=(u_0,u_1,\cdots,u_{\alpha-1}~|~v_0,v_1,\cdots,v_{\beta-1}~|~w_0,w_1,\cdots,w_{\gamma-1})$ maps to $\mathbf{w}(x)=(u_0+u_1x+\cdots+u_{\alpha-1}x^{\alpha-1}~|~v_0+v_1x+\cdots+v_{\beta-1}x^{\beta-1}~|~w_0+w_1x+\cdots+w_{\gamma-1}x^{\gamma-1}).$\par
 We can extend both the maps $\phi$ and $\epsilon$ to the polynomial rings $\mathbb{Z}_{p^s}[x]$ and $\mathbb{Z}_{p^r}[x]$ by applying these maps to each of the coefficients of the given polynomial. We define an external operation $*$ so that $R_{\alpha,\beta,\gamma}$ is a $\mathbb{Z}_{p^s}$-module. 
 \begin{definition}
 Define the operation $*:\mathbb{Z}_{p^s}[x]\times R_{\alpha,\beta,\gamma}\rightarrow R_{\alpha,\beta,\gamma} $ as follows
 \[\eta(x)*(p(x)~|~q(x)~|~r(x))= (\phi(\eta(x))p(x)~|~\phi(\eta(x))q(x)~|~\eta(x)r(x)),\]
 where $\eta(x)\in \mathbb{Z}_{p^s}$ and $(p(x)~|~q(x)~|~r(x))\in R_{\alpha,\beta,\gamma}$.
 \end{definition}
 With this external operation $*$ and usual addition, the ring $R_{\alpha,\beta,\gamma}$ is a $\mathbb{Z}_{p^s}$-module. Let $\mathbf{w}(x)=(u(x)~|~v(x)~|~w(x))$ be an element of $R_{\alpha,\beta,\gamma}$. If we operate $x$ on $\mathbf{w}(x)$, we get 
 \begin{align*}
     x*\mathbf{w}(x) &= x*(u(x)~|~v(x)~|~w(x))\\
     &= (u_0x+u_1x^2+\cdots+u_{\alpha-1}x^{\alpha}~|~v_0x+v_1x^2+\cdots+v_{\beta-1}x^{\beta}~|~w_0x+w_1x^2+\cdots+w_{\gamma-1}x^{\gamma})\\
     &= (u_{\alpha-1}+u_0x+\cdots+u_{\alpha-2}x^{\alpha-1}~|~v_{\beta-1}+v_0x+\cdots+v_{\beta-2}x^{\beta-1}~|~w_{\gamma-1}+w_0x+\cdots+w_{\gamma-2}x^{\gamma-1}).
     \end{align*}
     Hence, $x*\mathbf{w}(x)\in R_{\alpha,\beta,\gamma}$ and it corresponds to the vector $\mathbf{w}^{(1)}\in \mathbb{Z}_{p^r}\times\mathbb{Z}_{p^r}\times\mathbb{Z}_{p^s}$. Thus, if we operate $x$ on $\mathbf{w}(x)$ in $R_{\alpha,\beta,\gamma}$ corresponds to a shift of $\mathbf{w}$. In general case, $x^i*\mathbf{w}(x)= \mathbf{w}^{(i)}(x)$ for all $i$.\par
     From the above definition, we have the following  easily proven result.
     \begin{lemma}
     A $\mathbb{Z}_{p^r}\mathbb{Z}_{p^r}\mathbb{Z}_{p^s}$-additive code $C$ is a $\mathbb{Z}_{p^r}\mathbb{Z}_{p^r}\mathbb{Z}_{p^s}$-additive cyclic code of block length $(\alpha,\beta,\gamma)$ if and only if $C$ is a $\mathbb{Z}_{p^s}[x]$-submodule of $R_{\alpha,\beta,\gamma}$.
     \end{lemma}
    From the  structure of $\mathbb{Z}_4$-double cyclic codes studied in \cite{shi2105} and $\mathbb{Z}_{p^r}\mathbb{Z}_{p^s}$-additive cyclic codes discussed in \cite{BCV18}, we obtain the structure of double cyclic codes over $\mathbb{Z}_{p^r}$. The following theorem gives the generator polynomials of double cyclic codes over $\mathbb{Z}_{p^r}$ of block length $(\alpha,\beta)$. We consider $\gcd(\alpha, p)=\gcd(\beta,p)=1$.
     \begin{theorem}\label{th46}
     Let $\mathfrak{C}\subseteq \mathbb{Z}^{\alpha}_{p^r}\times \mathbb{Z}^{\beta}_{p^r}$ be a double cyclic code of block length $(\alpha,\beta)$. Then,
     \begin{equation*}
     \mathfrak{C}= \langle(a_0(x)+pa_1(x)+\cdots+p^{r-1}a_{r-1(x)}|0), (l(x)|b_0(x)+pb_1(x)+\cdots+p^{r-1}b_{r-1}(x))\rangle,
 \end{equation*}
 where $l(x)\in \mathbb{Z}_{p^r}[x]/\langle x^{\alpha}-1 \rangle$ and $a_{r-1}(x)|a_{r-2}(x)|\cdots|a_1(x)|a_0(x)|(x^{\alpha}-1)$,  $b_{r-1}(x)|b_{r-2}(x)|\cdots|b_1(x)|b_0(x)$ $|(x^{\beta}-1)$ over $\mathbb{Z}_{p^r}[x]$ .
     \end{theorem}
   \begin{proof}
   Proof follows from \cite[Proposition 2]{shi2105}.
   \end{proof} 
   
   A similar result can be given for the case when $\alpha=0$.
   \begin{theorem}
    Let $C\subseteq \mathbb{Z}^{\alpha}_{p^r}\times \mathbb{Z}^{\beta}_{p^r}\times \mathbb{Z}^{\gamma}_{p^s}$ be a $\mathbb{Z}_{p^r}\mathbb{Z}_{p^r}\mathbb{Z}_{p^s}$-additive cyclic code. If $\alpha=0$, then $C$ is a $\mathbb{Z}_{p^r}\mathbb{Z}_{p^s}$-additive cyclic code of block length $(\beta,\gamma)$ and
    \begin{equation*}
     C= \langle(b_0(x)+pb_1(x)+\cdots+p^{r-1}b_{r-1(x)}|0), (l_2(x)|g_0(x)+pg_1(x)+\cdots+p^{s-1}g_{s-1}(x))\rangle,
 \end{equation*}
 where $l_2(x)\in \mathbb{Z}_{p^r}[x]/\langle x^{\beta}-1 \rangle$ and $b_{r-1}(x)|b_{r-2}(x)|\cdots|b_1(x)|b_0(x)|(x^{\beta}-1)$ over $\mathbb{Z}_{p^r}[x]$ and  $g_{s-1}(x)|g_{s-2}(x)|$ $ \cdots|g_1(x)|g_0(x)|(x^{\gamma}-1)$ over $\mathbb{Z}_{p^s}[x]$.
   \end{theorem}
   \begin{proof}
   Proof follows from \cite[Theorem 4.3]{BCV18}.
   \end{proof}
\noindent Now we determine the generator polynomials of submodules of $R_{\alpha,\beta,\gamma}$. Now onward, we denote the $\mathbb{Z}_{p^s}[x]$-submodule generated by a subset $S$ of $R_{\alpha,\beta,\gamma}$ by $\langle S\rangle$.\par
 In rest of the paper, we consider $\mbox{gcd}(\alpha, p)=\mbox{gcd}(\beta, p)=\mbox{gcd}(\gamma, p)=1$ i.e. the integers $\alpha$, $\beta$ and $\gamma$ are coprime with $p$.  For this assumption, this has been proved in \cite{dinh, kanwar} that $\mathbb{Z}_{p^r}[x]/\langle x^{\alpha}-1\rangle$, $\mathbb{Z}_{p^r}[x]/\langle x^{\beta}-1\rangle$ and $\mathbb{Z}_{p^s}[x]/\langle x^{\gamma}-1\rangle$ are principal ideal rings.
 \begin{theorem}\label{th 4.4}
 Let $C$ be a $\mathbb{Z}_{p^s}[x]$-submodule of $R_{\alpha,\beta,\gamma}$. Then,
 \begin{equation*}
     C= \langle(\mathcal{A}(x)|0|0), (l(x)|\mathcal{B}(x)|0), (l_1(x)|l_2(x)|\mathcal{G}(x))\rangle,
 \end{equation*}
 where $\mathcal{A}(x)= a_0(x)+pa_1(x)+\cdots+p^{r-1}a_{r-1}(x)$ with $a_{r-1}(x)|a_{r-2}(x)|\cdots|a_1(x)|a_0(x)|(x^{\alpha}-1)$,  $\mathcal{B}(x)= b_0(x)+pb_1(x)+\cdots+p^{r-1}b_{r-1}(x)$ with $b_{r-1}(x)|b_{r-2}(x)|\cdots|b_1(x)|b_0(x)|(x^{\beta}-1)$ over $\mathbb{Z}_{p^r}[x]$,  $\mathcal{G}(x)= g_0(x)+pg_1(x)+\cdots+p^{s-1}g_{s-1}(x)$ with $g_{s-1}(x)|g_{s-2}(x)|\cdots|g_1(x)|g_0(x)|(x^{\gamma}-1)$ over $\mathbb{Z}_{p^s}[x]$ and $l(x),~ l_1(x)\in \mathbb{Z}_{p^r}[x]/ \langle x^{\alpha}-1\rangle$, $l_2(x)\in \mathbb{Z}_{p^r}[x]/ \langle x^{\beta}-1\rangle$.
 
 \end{theorem}
 \begin{proof}
 Since $C$ and $\mathbb{Z}_{p^s}[x]$ are $\mathbb{Z}_{p^s}[x]$-submodules of $R_{\alpha,\beta,\gamma}$, we define a map $\Phi: C \rightarrow \mathbb{Z}_{p^s}[x]/\langle x^{\gamma}-1\rangle$ such that $\Phi(u(x)|v(x)|w(x))= w(x)$. It can be seen that the map $\Phi$ is a $\mathbb{Z}_{p^s}[x]$-module homomorphism and the image $\Phi(C)$ is an ideal of $\mathbb{Z}_{p^s}[x]/\langle x^{\gamma}-1\rangle$.
 Since $\gamma$ and $p$ are coprime, from Theorem \ref{th2.4}, there exists a polynomial $\mathcal{G}(x)=g_0(x)+pg_1(x)+\cdots+p^{s-1}g_{s-1}(x)\in \mathbb{Z}_{p^s}[x]/\langle x^{\gamma}-1\rangle$, where $g_{s-1}(x)|g_{s-2}(x)|\cdots|g_1(x)|g_0(x)|(x^{\gamma}-1)$ such that $\Phi(C)=\langle \mathcal{G}(x)\rangle$. We have
 \[Ker(\Phi)=\{(u(x)|v(x)|0)\in R_{\alpha,\beta,\gamma}: (u(x)|v(x)|0)\in C\},\]
 We define the set 
 \[I= \{(u(x)|v(x))\in R_{\alpha,\beta}: (u(x)|v(x)|0)\in Ker(\Phi)\},\]
 where $R_{\alpha,\beta}= \mathbb{Z}_{p^r}[x]/\langle x^{\alpha}-1\rangle \times \mathbb{Z}_{p^r}[x]/\langle x^{\beta}-1\rangle$. It is clear that $I$ is a $\mathbb{Z}_{p^r}[x]$-submodule of $R_{\alpha,\beta}$. Therefore, $I$ is a double cyclic code over $\mathbb{Z}_{p^r}$. Thus, from Theorem \ref{th46}, $I$ can be written in the following form
 \[I= \langle (\mathcal{A}(x)|0), (l(x)|\mathcal{B}(x))\rangle,\]
 where $\mathcal{A}(x)= a_0(x)+pa_1(x)+\cdots+p^{r-1}a_{r-1}(x)$ with $a_{r-1}(x)|a_{r-2}(x)|\cdots|a_1(x)|a_0(x)|(x^{\alpha}-1)$ and $\mathcal{B}(x)= b_0(x)+pb_1(x)+\cdots+p^{r-1}b_{r-1}(x)$ with $b_{r-1}(x)|b_{r-2}(x)|\cdots|b_1(x)|b_0(x)|(x^{\beta}-1)$ and $l(x)\in \mathbb{Z}_{p^r}[x]/\langle x^{\alpha}-1\rangle$. This implies that 
 \[Ker(\Phi)= \langle (\mathcal{A}(x)|0|0), (l(x)|\mathcal{B}(x)|0)\rangle.\]
 On the other hand, by the First Isomorphism Theorem we have
 \[C/Ker(\Phi)\cong \Phi(C),\]
 which implies that
 \[C/Ker(\Phi)\cong \langle \mathcal{G}(x)\rangle.\]
 Let $(l_1(x)|l_2(x)|\mathcal{G}(x))\in C$ be such that
 \[\Phi(l_1(x)|l_2(x)|\mathcal{G}(x))= \mathcal{G}(x).\] Hence,  from this discussion we can conclude that any $\mathbb{Z}_{p^r}\mathbb{Z}_{p^r}\mathbb{Z}_{p^s}$-additive cyclic code can be generated as a $\mathbb{Z}_{p^s}[x]$-submodule of $R_{\alpha,\beta,\gamma}$ by the elements of the form   $(\mathcal{A}(x)|0|0), (l(x)|\mathcal{B}(x)|0)$ and $(l_1(x)|l_2(x)|\mathcal{G}(x))$. Moreover, any codeword of $C$ can be expressed as 
 \[m_1(x)*(\mathcal{A}(x)|0|0)+m_2(x)*(l(x)|\mathcal{B}(x)|0)+m_3(x)*(l_1(x)|l_2(x)|\mathcal{G}(x)),\]
where $m_1(x) ~\mbox{and}~ m_2(x)$ are elements in $\mathbb{Z}_{p^r}[x]$, $m_3(x)\in \mathbb{Z}_{p^s}[x]$.
 \end{proof}
 \begin{lemma}\label{lemma 4.5}
 Let $C=\langle (\mathcal{A}(x)|0|0),(l(x)|\mathcal{B}(x)|0), (l_1(x)|l_2(x)|\mathcal{G}(x))\rangle$ be a $\mathbb{Z}_{p^r}\mathbb{Z}_{p^r}\mathbb{Z}_{p^s}$-additive cyclic code. Then $\deg(l(x))<\deg(\mathcal{A}(x)),~\deg(l_1(x))<\deg(\mathcal A(x)),~ \deg(l_2(x))<\deg(\mathcal{B}(x))$, where $\deg(\mathcal{A}(x))= \deg(a_0(x))$, $\deg(\mathcal{B}(x))= \deg(b_0(x)).$
 \end{lemma}
 \begin{proof}
 Suppose that $\deg(l(x))\geq \deg(\mathcal{A}(x))$ with $\deg(l(x))- \deg(\mathcal{A}(x))=i$. Take \[D'=\langle(\mathcal{A}(x)|0|0),(l(x)+x^i\mathcal{A}(x)| \mathcal{B}(x)|0)\rangle.\] 
 Notice that \[(l(x)+x^i\mathcal{A}(x)|\mathcal{B}(x)|0)= x^i*(\mathcal{A}(x)|0|0)+ (l(x)|\mathcal{B}(x)|0),\]
 it follows that $D'\subseteq C$. On the other hand, we have \[(l(x)|\mathcal{B}(x)|0)=(l(x)+x^i\mathcal{A}(x)|\mathcal{B}(x)|0)- x^i*(\mathcal{A}(x)|0|0).\]
We get $C\subseteq D'$. Hence, $C=D'$. Therefore,  $\deg(l(x))<\deg(\mathcal{A}(x))$. Similarly, we can prove $\deg(l_1(x))< \deg(\mathcal A(x))$ and $\deg(l_2(x))< \deg(\mathcal{B}(x))$.
 \end{proof}
 Now, we present an example that illustrate the result obtained in Theorem \ref{th 4.4} and Lemma \ref{lemma 4.5}. 
 \begin{example}\label{ex:Z2Z2Z4}\em
Let $R_{3,3,5}=\mathbb Z_2[x]/\langle x^3-1\rangle\times \mathbb Z_2[x]/\langle x^3-1\rangle\times\mathbb Z_4[x]/\langle x^5-1\rangle$. Consider the $\mathbb Z_2\mathbb Z_2\mathbb Z_4$-additive cyclic code $C$ of block length $(3,3,5)$ generated by
$$
\{(1+x|1+x|0),(0|0|(1+x+x^2+x^3+x^4)+2)\},
$$
where $\mathcal{A}(x)=x^3-1,~l(x)=\mathcal{B}(x)=1+x,~ \mathcal{G}(x)=(1+x+x^2+x^3+x^4)+2$ and $l_1(x)=l_2(x)=0$. Note that every codeword $(\lambda_1(x)|\lambda_2(x)|\lambda_3(x))\in C$ satisfies that $\lambda_1(x)=\lambda_2(x)=\mu_1(x)(1+x)$ for some $\mu_1(x)\in\mathbb Z_2[x]/(x^3-1)$, $\deg(\mu_1(x))\leq 2$, and $\lambda_3(x)=\mu_2(x)((1+x+x^2+x^3+x^4)+2)$ for some $\mu_2(x)\in\mathbb Z_4[x]/(x^5-1)$, $\deg(\mu_2(x))\leq 1$.
\end{example}
\begin{remark}\label{R4.12}\em 
Let $C$ be a $\mathbb Z_{p^r}\mathbb Z_{p^r}\mathbb Z_{p^s}$-additive cyclic code of block length $(\alpha,\beta,\gamma)$. Note that when $p=2$ and $r=s=1$, we obtain triple cyclic codes over $\mathbb Z_2$ that are defined in \cite{mosta}. When $p=2$, $r=s=2$ then $C$ is a triple cyclic code over $\mathbb Z_4$ defined in \cite{Gao2016}. Finally, when $p=2$, $r=1$ and $s=2$, we have that $C$ is a $\mathbb Z_2\mathbb Z_2\mathbb Z_4$- additive cyclic code studied in \cite{224}. The generator polynomials of the code $C$ given in Theorem \ref{th 4.4} is therefore a generalization of the statements that provide the generator polynomials of triple cyclic codes over $\mathbb Z_2$ \cite[Theorem 3.2]{mosta}, triple cyclic codes over $\mathbb Z_4$ \cite[Proposition 2.4]{Gao2016} and $\mathbb Z_2\mathbb Z_2\mathbb Z_4$-additive cyclic code
 \cite[Theorem 3.1]{224}. However, in all these statements, the polynomial $l(x)$ is zero. Therefore, these statements do not give the generator polynomial  of a general additive cyclic code but only of those codes with $l(x)=0$. For instance, the $\mathbb Z_2\mathbb Z_2\mathbb Z_4$-additive cyclic code given in Example \ref{ex:Z2Z2Z4}, can not be generated with the generator polynomials described in \cite{224}.
 \end{remark}

 \begin{lemma}\label{lemma4.6}
 Let $C= \langle (\mathcal{A}(x)|0|0),(l(x)|\mathcal{B}(x)|0), (l_1(x)|l_2(x)|\mathcal{G}(x))\rangle$ be a $\mathbb{Z}_{p^r}\mathbb{Z}_{p^r}\mathbb{Z}_{p^s}$-additive cyclic code of block length $(\alpha,\beta,\gamma)$. Then,
 \begin{enumerate}[label=(\roman*)]
     \item $\mathcal{A}(x)|\phi\Big(\frac{x^{\beta}-1}{b_{r-1}(x)}\Big)l(x)$.
     
     \item $\mathcal{B}(x)|\phi\Big(\frac{x^{\gamma}-1}{g_{s-1}(x)}\Big)l_2(x)$.%and $l(x)|\phi\Big(\frac{x^{\gamma}-1}{g_{s-1}(x)}\Big)l_1(x)$.
     
     \item $\mathcal{A}(x)| (Q(x)l(x)- \phi\Big(\frac{x^{\gamma}-1}{g_{s-1}(x)}\Big)l_1(x))$, where $Q(x)\mathcal{B}(x)= \phi\Big(\frac{x^{\gamma}-1}{g_{s-1}(x)}\Big)l_2(x).$
 \end{enumerate}
 \end{lemma}
 \begin{proof}
 \begin{enumerate}[label=(\roman*)]
     \item We consider
     \begin{align*}
         \prod_{t=0}^{r-1}\hat{b}_t*(l(x)|\mathcal{B}(x)|0)&= \frac{x^{\beta}-1}{b_{r-1}(x)}*(l(x)|\mathcal{B}(x)|0)\\
         &= \Big(\phi\Big(\frac{x^{\beta}-1}{b_{r-1}(x)}\Big)l(x)|\frac{x^{\beta}-1}{b_{r-1}(x)}\mathcal{B}(x)|0\Big)\\
         &= \Big(\phi\Big(\frac{x^{\beta}-1}{b_{r-1}(x)}\Big)l(x)|0|0\Big).
     \end{align*}
     Therefore, $\Big(\phi\Big(\frac{x^{\beta}-1}{b_{r-1}(x)}\Big)l(x)|0|0\Big)\in \langle (\mathcal{A}(x)|0|0)\rangle$. Hence, $\mathcal{A}(x)|\phi\Big(\frac{x^{\beta}-1}{b_{r-1}(x)}\Big)l(x)$.
     
    \item Let us consider
    \begin{align*}
        \prod_{t=0}^{s-1}\hat{g}_t*(l_1(x)|l_2(x)|\mathcal{G}(x))&= \frac{x^{\gamma}-1}{g_{s-1}(x)}*(l_1(x)|l_2(x)|\mathcal{G}(x))\\
        &=  \Big(\phi\Big(\frac{x^{\gamma}-1}{g_{s-1}(x)}\Big)l_1(x)|\phi\Big(\frac{x^{\gamma}-1}{g_{s-1}(x)}\Big)l_2(x)|0\Big).
    \end{align*}
    Since 
    \[\Phi\Big(\phi\Big(\frac{x^{\gamma}-1}{g_{s-1}(x)}\Big)l_1(x)|\phi\Big(\frac{x^{\gamma}-1}{g_{s-1}(x)}\Big)l_2(x)|0\Big)= 0,\]
    therefore,
    \[\Big(\phi\Big(\frac{x^{\gamma}-1}{g_{s-1}(x)}\Big)l_1(x)|\phi\Big(\frac{x^{\gamma}-1}{g_{s-1}(x)}\Big)l_2(x)|0\Big)\in Ker(\Phi)\subseteq C.\]
    Hence, we get $\mathcal{B}(x)|\phi\Big(\frac{x^{\gamma}-1}{g_{s-1}(x)}\Big)l_2(x)$.
    
    \item On the one hand, since $\mathcal{B}(x)|\phi\Big(\frac{x^{\gamma}-1}{g_{s-1}(x)}\Big)l_2(x),$ then there exists a $Q(x)\in \mathbb{Z}_{p^r}[x]$ such that
    \[\mathcal{B}(x)Q(x)=\phi\Big(\frac{x^{\gamma}-1}{g_{s-1}(x)}\Big)l_2(x).\]
    Let us consider
    \[Q(x)*(l(x)|\mathcal{B}(x)|0)= (Q(x)l(x)|Q(x)\mathcal{B}(x)|0)\in C.\]
    On the other hand, we have 
    \[\Big(\phi\Big(\frac{x^{\gamma}-1}{g_{s-1}(x)}\Big)l_1(x)|\phi\Big(\frac{x^{\gamma}-1}{g_{s-1}(x)}\Big)l_2(x)|0\Big)\in C.\]
    So, we get
    \begin{align*}
       (Q(x)l(x)|Q(x)\mathcal{B}(x)|0)&- \Big(\phi\Big(\frac{x^{\gamma}-1}{g_{s-1}(x)}\Big)l_1(x)|\phi\Big(\frac{x^{\gamma}-1}{g_{s-1}(x)}\Big)l_2(x)|0\Big)\\
       &= \Big(Q(x)l(x)-\phi\Big(\frac{x^{\gamma}-1}{g_{s-1}(x)}\Big)l_1|0|0\Big)\in Ker(\Phi)\subseteq C.
    \end{align*}
    Therefore, $\mathcal{A}(x)|(Q(x)l(x)- \phi\Big(\frac{x^{\gamma}-1}{g_{s-1}(x)}\Big)l_1(x)).$
    
 \end{enumerate}
 \end{proof}
 
 If any $\mathbb{Z}_{p^r}\mathbb{Z}_{p^r}\mathbb{Z}_{p^s}$-additive cyclic code has only one generator $(l_1(x)|l_2(x)|\mathcal{G}(x))$, then from Lemma \ref{lemma4.6}, we have  $(x^{\alpha}-1)|\phi\Big(\frac{x^{\gamma}-1}{g_{s-1}(x)}\Big)l_1(x),~ (x^{\beta}-1)|\phi\Big(\frac{x^{\gamma}-1}{g_{s-1}(x)}\Big)l_2(x)$ and $\mathcal{G}(x)$ holds the conditions given in Theorem \ref{th 4.4}.\par
 From Lemma \ref{lemma 4.5}  and \ref{lemma4.6}, we have the following result.
 \begin{corollary}
 If $\gcd\Big(\mathcal{A}(x), \phi\Big(\frac{x^{\beta}-1}{b_{r-1}(x)}\Big)\Big)= \gcd\Big(\mathcal{B}(x), \phi\Big(\frac{x^{\gamma}-1}{g_{s-1}(x)}\Big)\Big)=1$, then $l(x)=0,~ l_2(x)=0.$
 \end{corollary}
 
 In the above discussion, we have determined the generator polynomials of $\mathbb{Z}_{p^r}\mathbb{Z}_{p^r}\mathbb{Z}_{p^s}$-additive cyclic codes. Now, from these generator polynomials, we determine the generator polynomials of separable codes.
 \begin{lemma}\label{lemma4.8}
 Let $C=\langle (\mathcal{A}(x)|0|0),(l(x)|\mathcal{B}(x)|0), (l_1(x)|l_2(x)|\mathcal{G}(x))\rangle$ be a $\mathbb{Z}_{p^r}\mathbb{Z}_{p^r}\mathbb{Z}_{p^s}$-additive cyclic code of block length $(\alpha,\beta,\gamma)$. Then
 \[C_{\alpha}= \langle \gcd(\mathcal{A}(x),l(x), l_1(x))\rangle,~ C_{\beta}= \langle \gcd(\mathcal{B}(x),l_2(x))\rangle,~ C_{\gamma}= \langle \mathcal{G}(x)\rangle. \]
 \end{lemma}
 \begin{proof}
 On the one hand, let $u(x)\in C_{\alpha}$, then there exist $v(x)\in \mathbb{Z}_{p^r}[x]/\langle x^{\beta}-1\rangle,~ w(x)\in \mathbb{Z}_{p^s}[x]/\langle x^{\gamma}-1\rangle$ such that $(u(x)|v(x)|w(x))\in C$. It follows that there exist some polynomials $m_1(x),m_2(x),m_3(x)\in \mathbb{Z}_{p^s}[x]$ such that
 \[(u(x)|v(x)|w(x))= m_1(x)*(\mathcal{A}(x)|0|0)+m_2(x)*(l(x)|\mathcal{B}(x)|0)+m_3(x)*(l_1(x))|l_2(x)|\mathcal{G}(x)).\]
 Clearly, \[u(x)= \phi(m_1(x))\mathcal{A}(x)+ \phi(m_2(x))l(x)+ \phi(m_3(x))l_1(x).\]
 This implies that $\mbox{gcd}(\mathcal{A}(x), l(x), l_1(x))|u(x)$. Thus, $u(x)\in \langle\mbox{gcd}(\mathcal{A}(x), l(x), l_1(x))\rangle$. It follows that $C_{\alpha}\subseteq \langle\mbox{gcd}(\mathcal{A}(x),l(x),l_1(x))\rangle.$ On the other hand, there exist $\lambda_1(x),\lambda_2(x), \lambda_3(x)\in \mathbb{Z}_{p^r}[x]$ such that 
 \begin{align*}
   \mbox{gcd}(\mathcal{A}(x),l(x), l_1(x))&= \lambda_1(x)\mathcal{A}(x)+\lambda_2(x)l(x)+\lambda_3(x)l_1(x)\\
   &= \phi(\lambda_1(x))\mathcal{A}(x)+\phi(\lambda_2(x))l(x)+\lambda_3(x)l_1(x).\\
    &~~(\mbox{Since}~ \lambda_1(x),\lambda_2(x),\lambda_3(x)\in \mathbb{Z}_{p^r}[x], ~\mbox{then}~ \phi(\lambda_i(x))=\lambda_i(x),~\mbox{for}~ i=1,2,3.)
 \end{align*}
 Hence, we have that
 \begin{align*}
     &(\mbox{gcd}(\mathcal{A}(x),l(x),l_1(x))|\lambda_2\mathcal{B}(x)+\lambda_3(x)l_2(x)|\lambda_3\mathcal{G}(x))\\&= \lambda_1(x)*(\mathcal{A}(x)|0|0)+\lambda_2(x)*(l(x)|\mathcal{B}(x)|0)+\lambda_3(x)*(l_1(x)|l_2(x)|\mathcal{G}(x))\in C.
    \end{align*}
It follows that $\langle\mbox{gcd}(\mathcal{A}(x),l(x), l_1(x))\rangle\subseteq C_{\alpha}.$ Thus, we get $C_{\alpha}= \langle\mbox{gcd}(\mathcal{A}(x), l(x),l_1(x))\rangle$.
Similarly, we can prove that $C_{\beta}= \langle \mbox{gcd}(\mathcal{B}(x),l_2(x))\rangle,~ C_{\gamma}= \langle \mathcal{G}(x)\rangle.$
 \end{proof}
 \begin{lemma}\label{lemma4.9}
 Let $C=\langle (\mathcal{A}(x)|0|0),(l(x)|\mathcal{B}(x)|0), (l_1(x)|l_2(x)|\mathcal{G}(x))\rangle$ be a $\mathbb{Z}_{p^r}\mathbb{Z}_{p^r}\mathbb{Z}_{p^s}$-additive cyclic code of block length $(\alpha,\beta,\gamma)$. Then $\mathcal{A}(x)|l(x)$ if and only if $l(x)=0$ and $\mathcal{A}(x)|l_1(x)$ if and only if $l_1(x)=0$.
 \end{lemma}
 \begin{proof}
 If $l(x)=0$, then clearly $\mathcal{A}(x)|l(x)$.\par
Suppose that $\mathcal{A}(x)|l(x)$, then there exists a polynomial $\mu(x)\in \mathbb{Z}_{p^r}[x]$ such that $l(x)= \mu(x)\mathcal{A}(x)$. Let
 \[C'= \langle (\mathcal{A}(x)|0|0),(0|\mathcal{B}(x)|0), (l_1(x)|l_2(x)|\mathcal{G}(x))\rangle.\]
 On the one hand, notice that 
 \[(0|\mathcal{B}(x)|0)=(l(x)|\mathcal{B}(x)|0)- \mu(x)*(\mathcal{A}(x)|0|0)\in C,\]
 thus, $C'\subseteq C$. On the other hand,
 \[(l(x)|\mathcal{B}(x)|0)= \mu(x)*(\mathcal{A}(x)|0|0)+(0|\mathcal{B}(x)|0)\in C',\]
 thus, $C\subseteq C'$. Therefore, $C=C'$, which implies $l(x)=0.$ Similarly, we can prove $\mathcal{A}(x)|l_1(x)$ if and only if $l_1(x)=0$.
 \end{proof}
 Similar to the above lemma, we have one more result as follows.
 \begin{lemma}\label{lemma4.10}
  Let $C=\langle (\mathcal{A}(x)|0|0),(l(x)|\mathcal{B}(x)|0), (l_1(x)|l_2(x)|\mathcal{G}(x))\rangle$ be a $\mathbb{Z}_{p^r}\mathbb{Z}_{p^r}\mathbb{Z}_{p^s}$-additive cyclic code of block length $(\alpha,\beta,\gamma)$. Then $\mathcal{B}(x)|l_2(x)$ if and only if $l_2(x)=0$.
 \end{lemma}
By using the results discussed in Lemmas \ref{lemma4.8}, \ref{lemma4.9}, \ref{lemma4.10}, we get the following results for a $\mathbb{Z}_{p^r}\mathbb{Z}_{p^r}\mathbb{Z}_{p^s}$-additive cyclic code to be a separable code.
 \begin{lemma}
  Let $C=\langle (\mathcal{A}(x)|0|0),(l(x)|\mathcal{B}(x)|0), (l_1(x)|l_2(x)|\mathcal{G}(x))\rangle$ be a $\mathbb{Z}_{p^r}\mathbb{Z}_{p^r}\mathbb{Z}_{p^s}$-additive cyclic code of block length $(\alpha,\beta,\gamma)$. Then the following statements are equivalent:
  \begin{enumerate}
      \item C is separable;
      \item $\mathcal{A}(x)|l(x),~ \mathcal{A}(x)|l_1(x)$, $\mathcal{B}(x)|l_2(x)$;
      \item $C_{\alpha}= \langle \mathcal{A}(x)\rangle$, $C_{\beta}= \langle \mathcal{B}(x)\rangle$;
      \item $C=\langle (\mathcal{A}(x)|0|0),(0|\mathcal{B}(x)|0), (0|0|\mathcal{G}(x))\rangle$.
  \end{enumerate}
 \end{lemma}
 \begin{proof}
 $(1)\implies (2):$\\
 If $C$ is separable, then by Lemma \ref{lemma4.8}, we have
 \[C=C_{\alpha}\times C_{\beta}\times C_{\gamma}= \langle \mbox{gcd}(\mathcal{A}(x),l(x),l_1(x))\rangle\times \langle \mbox{gcd}(\mathcal{B}(x),l_2(x))\rangle\times \langle \mathcal{G}(x)\rangle.\]
 Since $(\mbox{gcd}(\mathcal{A}(x),l(x),l_1(x))|0|0)\in C$, this implies that there exists some $\lambda(x)\in \mathbb{Z}_{p^r}[x]$ such that $\mbox{gcd}(\mathcal{A}(x),l(x),l_1(x))= \lambda(x)\mathcal{A}(x)$. Hence, $\mathcal{A}(x)|l(x), ~\mathcal{A}(x)|l_1(x)$. Similarly, we can obtain that $\mathcal{B}(x)|l_1(x).$\\
 $(2)\iff (3):$\\
 From Lemma \ref{lemma4.8}, we can obtain this result.\\
 $(2)\implies (4)$:\\
 From previous Lemmas \ref{lemma4.9} and \ref{lemma4.10}, we can prove this result.\\
 $(4)\implies (1):$\\
 Suppose that, $C=\langle (\mathcal{A}(x)|0|0),(0|\mathcal{B}(x)|0), (0|0|\mathcal{G}(x))\rangle$, this implies that $C= \langle \mathcal{A}(x)\rangle \times \langle \mathcal{B}(x)\rangle \times \langle \mathcal{G}(x)\rangle = C_{\alpha}\times C_{\beta}\times C_{\gamma}.$ Hence, $C$  is separable.
 \end{proof}
 In the above discussions, We studied the generator polynomials of $\mathbb{Z}_{p^r}\mathbb{Z}_{p^r}\mathbb{Z}_{p^s}$- additive cyclic codes and each of their canonical projection. We also have discussed some necessary and sufficient conditions on these polynomials. These conditions are used to determine generator polynomials of separable codes. We now classify some cases of $\mathbb{Z}_{p^r}\mathbb{Z}_{p^r}\mathbb{Z}_{p^s}$- additive cyclic codes as follows.
\begin{theorem}\label{th4.12}
Let $C$ be a $\mathbb{Z}_{p^r}\mathbb{Z}_{p^r}\mathbb{Z}_{p^s}$- additive cyclic code of block length $(\alpha,\beta,\gamma)$. Then the classification of $C$ is as follows:
\begin{enumerate}
\item $C=\langle \mathcal{A}(x)|0|0\rangle$ where $\mathcal{A}(x)=a_0(x)+pa_1(x)+\cdots+p^{r-1}a_{r-1}(x)$ with $a_{r-1}(x)|a_{r-2}(x)|\cdots|a_1(x)|$ $a_0(x)|(x^{\alpha}-1)$;
\item $C=\langle0|\mathcal{B}(x)|0\rangle$ where $\mathcal{B}(x)=b_0(x)+pb_1(x)+\cdots+p^{r-1}b_{r-1}(x)$ with $b_{r-1}(x)|b_{r-2}(x)|\cdots|b_1(x)|$ $b_0(x)|(x^{\beta}-1)$;
\item $C=\langle0|0|\mathcal{G}(x)\rangle$ where $\mathcal{G}(x)=g_0(x)+pg_1(x)+\cdots+p^{s-1}g_{s-1}(x)$ with $g_{s-1}(x)|g_{s-2}(x)|\cdots|g_1(x)|$ $g_0(x)|(x^{\gamma}-1)$;
\item $C=\langle l(x)|\mathcal{B}(x)|0\rangle$ where $\mathcal{B}(x)=b_0(x)+pb_1(x)+\cdots+p^{r-1}b_{r-1}(x)$ with $b_{r-1}(x)|b_{r-2}(x)|\cdots|b_1(x)|$ $b_0(x)|(x^{\beta}-1)$ and $(x^{\alpha}-1)|\phi\Big(\frac{x^{\beta}-1}{b_{r-1}(x)}\Big)l(x)$;

\item $C=\langle (l_1(x)|l_2(x)|\mathcal{G}(x))\rangle$ with $(x^{\alpha}-1)|\phi\big(\frac{x^{\gamma}-1}{g_{s-1}(x)}\big)l_1(x), ~(x^{\beta}-1)|\phi\big(\frac{x^{\gamma}-1}{g_{s-1}(x)}\big)l_2(x)$ and the polynomial $\mathcal{G}(x)=g_0(x)+pg_1(x)+\cdots+p^{s-1}g_{s-1}(x)$ satisfies $g_{s-1}(x)|g_{s-2}(x)|\cdots|g_1(x)|g_0(x)|(x^{\gamma}-1)$;

\item $C=\langle (\mathcal{A}(x)|0|0),(l(x)|\mathcal{B}(x)|0) \rangle$ with $\deg (l(x))< \deg (\mathcal{A}(x)),~ \mathcal{A}(x)|\phi\Big(\frac{x^{\beta}-1}{b_{r-1}(x)}\Big)l(x)$ and the polynomials $\mathcal{A}(x),~\mathcal{B}(x)$ satisfy above mentioned conditions;
\item $C=\langle (\mathcal{A}(x)|0|0),(l_1(x)|l_2(x)|\mathcal{G}(x)\rangle$ with $\deg(l_1(x))< \deg(\mathcal{A}(x)),~\mathcal{A}(x)|\phi\Big(\frac{x^{\gamma}-1}{g_{s-1}(x)}\Big)l_1(x)$, $(x^{\beta}-1)|\phi\big(\frac{x^{\gamma}-1}{g_{s-1}(x)}\big)l_2(x)$ and the polynomials $\mathcal{A}(x),~\mathcal{G}(x)$ satisfy above mentioned conditions;
\item $C=\langle (l(x)|\mathcal{B}(x)|0),(l_1(x)|l_2(x)|\mathcal{G}(x)\rangle$ with $(x^{\alpha}-1)|\phi\big(\frac{x^{\beta}-1}{b_{r-1}(x)}\big)l(x),~(x^{\alpha}-1)|\phi\big(\frac{x^{\gamma}-1}{g_{s-1}(x)}\big)l_1(x)$, $ $ $\deg(l_2(x))< \deg(\mathcal{B}(x))$ and the polynomials $\mathcal{B}(x),~\mathcal{G}(x)$ satisfy above mentioned conditions;
\item $C=\langle (0|\mathcal{B}(x)|0),(0|l_2(x)|\mathcal{G}(x)\rangle$ with $\mathcal{B}(x)|\phi\big(\frac{x^{\gamma}-1}{g_{s-1}(x)}\big)l_2(x)$, $\deg(l_2(x))< \deg(\mathcal{B}(x))$ and the polynomials $\mathcal{B}(x),~\mathcal{G}(x)$ satisfy above mentioned conditions;
\item $C=\langle (\mathcal{A}(x)|0|0), (0|\mathcal{B}(x)|0),(l_1(x)|l_2(x)|\mathcal{G}(x)\rangle$ with  $\deg(l_1(x))< \deg(\mathcal{A}(x)),~\mathcal{A}(x)|\phi\Big(\frac{x^{\gamma}-1}{g_{s-1}(x)}\Big)l_1(x)$, $\deg(l_2(x))< \deg(\mathcal{B}(x)),~\mathcal{B}(x)|\phi\big(\frac{x^{\gamma}-1}{g_{s-1}(x)}\big)l_2(x)$ and the polynomials $\mathcal{A}(x),~\mathcal{B}(x),~\mathcal{G}(x)$ satisfy above mentioned conditions;
\item $C=\langle (\mathcal{A}(x)|0|0), (0|\mathcal{B}(x)|0),(0|0|\mathcal{G}(x)\rangle$ and the polynomials $\mathcal{A}(x),~\mathcal{B}(x),~\mathcal{G}(x)$ satisfy above mentioned conditions;
\item $C=\langle (\mathcal{A}(x)|0|0), (l(x)|\mathcal{B}(x)|0),(l_1(x)|0|\mathcal{G}(x)\rangle$ with  $\deg(l_1(x))< \deg(\mathcal{A}(x)),~\deg(l_1(x))< \deg(\mathcal{A}(x))$, $\mathcal{A}(x)|\phi\Big(\frac{x^{\gamma}-1}{g_{s-1}(x)}\Big)l_1(x),~\mathcal{A}(x)|\phi\Big(\frac{x^{\beta}-1}{b_{r-1}(x)}\Big)l(x)$and the polynomials $\mathcal{A}(x),~\mathcal{B}(x),~\mathcal{G}(x)$ satisfy above mentioned conditions;
\item $C=\langle (\mathcal{A}(x)|0|0), (l(x)|\mathcal{B}(x)|0),(0|l_2(x)|\mathcal{G}(x)\rangle$ with  $\deg(l(x))< \deg(\mathcal{A}(x)),~\mathcal{A}(x)|\phi\Big(\frac{x^{\beta}-1}{b_{r-1}(x)}\Big)l(x)$, $\deg(l_2(x))< \deg(\mathcal{B}(x)),~\mathcal{B}(x)|\phi\big(\frac{x^{\gamma}-1}{g_{s-1}(x)}\big)l_2(x)$ and the polynomials $\mathcal{A}(x),~\mathcal{B}(x),~\mathcal{G}(x)$ satisfy above mentioned conditions;
\item $C=\langle (\mathcal{A}(x)|0|0), (l(x)|\mathcal{B}(x)|0),(l_1(x)|l_2(x)|\mathcal{G}(x)\rangle$ with  $\deg(l(x))< \deg(\mathcal{A}(x))$, $\deg(l_1(x))<\linebreak \deg(\mathcal{A}(x))$, $\deg(l_2(x))< \deg(\mathcal{B}(x))$, where $l(x),~l_1(x),~l_2(x)$ satisfy the conditions presented in Lemma \ref{lemma4.6},  and the polynomials $\mathcal{A}(x),~\mathcal{B}(x),~\mathcal{G}(x)$ satisfy above mentioned conditions.
\end{enumerate}
\end{theorem}
We obtain the following remark from Theorem \ref{th4.12}.
\begin{remark}\em
 In the cases $(1)$, $(2)$ and $(3)$, the code $C$ is a $\mathbb{Z}_{p^r}$-additive cyclic code,  $\mathbb{Z}_{p^s}$-additive cyclic code and in the case $(6)$ the code $C$ is a double cyclic code over $\mathbb{Z}_{p^r}$. In the case $(9)$ the code $C$ is a $\mathbb{Z}_{p^r}\mathbb{Z}_{p^s}$-additive cyclic code. Finally, the case $(14)$ is the general definition of $\mathbb{Z}_{p^r}\mathbb{Z}_{p^r}\mathbb{Z}_{p^s}$-additive cyclic codes that covers all the other cases.
\end{remark}
Now, we present an example to illustrate our results.
\begin{example}\em
 Let $R_{5,5,2}=\mathbb{Z}_3[x]/\langle x^5-1\rangle\times\mathbb{Z}_3[x]/\langle x^5-1\rangle\times\mathbb{Z}_9[x]/\langle x^2-1\rangle$. Consider a $\mathbb{Z}_3\mathbb{Z}_3\mathbb{Z}_9$-additive cyclic code $C$ of block length $(5,5,2)$  generated by
 $$
 \{ (1+x+x^2+x^3+x^4|0|0),(1+x|2+x|0),(1|0|(1+x)+3)\}.
 $$
 Then, $a_0(x)=1+x+x^2+x^3+x^4,~b_0(x)=2+x,~l(x)=1+x,~l_1(x)=1,~l_2(x)=0,~g_0(x)=1+x$ and $g_1(x)=1$.
 \end{example}
 \section{Minimal Generating Sets}
In this section, our goal is to find minimal generating sets of $\mathbb{Z}_{p^r}\mathbb{Z}_{p^r}\mathbb{Z}_{p^s}$-additive cyclic codes as  $\mathbb{Z}_{p^s}$-submodules. Once we found this set, we will use it to determine the size of $\mathbb{Z}_{p^s}$-submodules  of $R_{\alpha,\beta,\gamma}$.\par
For the rest of the discussion, any $\mathbb{Z}_{p^r}\mathbb{Z}_{p^r}\mathbb{Z}_{p^s}$-additive cyclic code $C$ of block length $(\alpha,\beta,\gamma)$ is of the from $C= \langle (\mathcal{A}(x)|0|0),(l(x)|\mathcal{B}(x)|0), (l_1(x)|l_2(x)|\mathcal{G}(x))\rangle$ where $\mathcal{A}(x)= a_0(x)+pa_1(x)+\cdots+p^{r-1}a_{r-1}(x),~ \mathcal{B}(x)=b_0(x)+pb_1(x)+\cdots+p^{r-1}b_{r-1}(x)$ and $\mathcal{G}(x)=g_0(x)+pg_1(x)+\cdots+p^{s-1}g_{s-1}(x)$, for the polynomials $a_i(x),~b_j(x),~\mbox{and}~ g_k(x)$ as in Theorem \ref{th 4.4}. Since $a_0(x)$ is a factor of $x^{\alpha}-1$ and the polynomial $a_i(x)$ is a factor of $a_{i-1}(x)$ for $i=1,2,\cdots, r-1$, we denote $\hat{a}_0(x)= \frac{x^{\alpha}-1}{a_0(x)},~ \hat{a}_i(x)= \frac{a_{i-1}(x)}{a_i(x)}$ for $i=1,2,\cdots, r-1$ and $\hat{a}_r(x)= a_{r-1}(x)$. Similarly, we define $\hat{b}_0(x)= \frac{x^{\beta}-1}{b_0(x)},~ \hat{b}_j(x)= \frac{b_{j-1}(x)}{b_j(x)}$ for $j=1,2,\cdots, r-1$ and $\hat{b}_r(x)= b_{r-1}(x)$. In the same way, we define $\hat{g}_0(x)= \frac{x^{\gamma}-1}{g_0(x)},~ \hat{g}_k(x)= \frac{g_{k-1}(x)}{g_k(x)}$ for $k=1,2,\cdots, s-1$ and $\hat{g}_s(x)= g_{s-1}(x)$.
\begin{theorem}\label{th5.1}
 Let $C=\langle (\mathcal{A}(x)|0|0), (l(x)|\mathcal{B}(x)|0),(l_1(x)|l_2(x)|\mathcal{G}(x))$ be a $\mathbb{Z}_{p^r} \mathbb{Z}_{p^r} \mathbb{Z}_{p^s}$-additive cyclic code of block length $(\alpha,\beta,\gamma)$, where $\mathcal{A}(x)=a_0(x)+pa_1(x)+\cdots+p^{r-1}a_{r-1}(x) ,~\mathcal{B}(x)=b_0(x)+pb_1(x)+\cdots+p^{r-1}b_{r-1}(x),~\mathcal{G}(x)=g_0(x)+pg_1(x)+\cdots+p^{s-1}g_{s-1}(x)$ such that $a_{r-1}(x)|a_{r-2}(x)|\cdots|a_1(x)|a_0(x)|(x^{\alpha}-1)$, $b_{r-1}(x)|b_{r-2}(x)|\cdots|b_1(x)|$ $b_0(x)|(x^{\beta}-1),~g_{s-1}(x)|g_{s-2}(x)|\cdots|g_1(x)|g_0(x)|(x^{\gamma}-1).$ Define, 
 \[A_i= \Bigg\{x^m\big(\prod_{t=0}^{i-1}\hat{a}_t(x)\big)*(\mathcal{A}(x)|0|0)\Bigg\}_{m=0}^{\deg(\hat{a}_i(x))-1}\]
 for $0\leq i\leq r-1,$
 \[B_j= \Bigg\{x^m\big(\prod_{t=0}^{j-1}\hat{b}_t(x)\big)*(l(x)|\mathcal{B}(x)|0)\Bigg\}_{m=0}^{\deg(\hat{b}_j(x))-1}\]
 for $0\leq j\leq r-1,$ and
  \[G_k= \Bigg\{x^m\big(\prod_{t=0}^{k-1}\hat{g}_t(x)\big)*(l_1(x)|l_2(x)|\mathcal{G}(x))\Bigg\}_{m=0}^{\deg(\hat{g}_k(x))-1}\]
 for $0\leq k\leq s-1$. Then
 \[S=\bigg(\bigcup_{i=0}^{r-1}A_i\bigg)\cup\bigg(\bigcup_{j=0}^{r-1}B_j\bigg)\cup \bigg(\bigcup_{k=0}^{s-1}G_k\bigg)\]
 forms a minimal generating set for $C$ as a $\mathbb{Z}_{p^s}$-submodule. Furthermore, 
 \[|C|=p^{\sum_{i=0}^{r-1}(r-i)\deg(\hat{a}_i(x))+\sum_{j=0}^{r-1}(r-j)\deg(\hat{b}_j(x))+\sum_{k=0}^{s-1}(s-k)\deg(\hat{g}_k)}.\]
\end{theorem}
 \begin{proof}
By Theorem \ref{th2.5}, we can easily see that the elements in $S$ are $\mathbb{Z}_{p^s}$-linearly independent as $\big(\bigcup_{i=0}^{r-1}A_i\big)_{\alpha}$, $\big(\bigcup_{j=0}^{r-1}B_j\big)_{\beta}$ and $\big(\bigcup_{k=0}^{s-1}G_k\big)_{\gamma}$ form  minimal generating sets for the codes $C_{\alpha}$, $C_{\beta}$ and $C_{\gamma}$, respectively. Let $c(x)$ be any codeword of $C$, then \[c(x)=m_1(x)*(\mathcal{A}(x)|0|0)+m_2(x)* (l(x)|\mathcal{B}(x)|0)+m_3(x)*(l_1(x)|l_2(x)|\mathcal{G}(x)).\] 
By applying similar argument as Theorem \ref{th2.5}, we get \[ m_1(x)*(\mathcal{A}(x)|0|0)\in \langle\cup_{i=0}^{r-1}A_i\rangle_{\mathbb{Z}_{p^s}}.\] 
Now, if $\deg (m_2(x))< \deg (\hat{b}_0(x))$ then $m_2(x)* (l(x)|\mathcal{B}(x)|0)\in \langle B_0\rangle_{\mathbb{Z}_{p^s}}.$ Otherwise, by division algorithm $m_2(x)= q_0(x)\hat{b}_0(x)+p_0(x)$ with $\mbox{deg}(p_0(x))<\mbox{deg}(\hat{b}_0(x))$. Then, \[m_2(x)*(l(x)|\mathcal{B}(x)|0)=q_0(x)\hat{b}_0(x)*(l(x)|\mathcal{B}(x)|0)+p_0(x)*(l(x)|\mathcal{B}(x)|0).\] 
Clearly, $p_0(x)*(l(x)|\mathcal{B}(x)|0)\in \langle B_0\rangle_{\mathbb{Z}_{p^s}}$. If $\deg (q_0(x))< \deg (\hat{b}_1(x))$, then $q_0(x)\hat{b}_0(x)*(l(x)|\mathcal{B}(x)|0)\in \langle B_1\rangle_{\mathbb{Z}_{p^s}}$.\par
In the worst-case scenario and by applying similar arguments, one can obtain that \[m_2(x)*(l(x)|\mathcal{B}(x)|0)\in \langle S\rangle_{\mathbb{Z}_{p^s}},\] 
if $q_{r-2}(x)(\prod_{t=0}^{r-2}\hat{b}_t(x))*(l(x)|\mathcal{B}(x)|0)\in \langle S \rangle_{\mathbb{Z}_{p^s}}.$ It is clear, if $\mbox{deg}(q_{r-2}(x))< \mbox{deg}(\hat{b}_{r-1}(x))$ then \[q_{r-2}(x)\Big(\prod_{t=0}^{r-2}\hat{b}_t(x)\Big)*(l(x)|\mathcal{B}(x)|0)\in \langle B_{r-1}\rangle_{\mathbb{Z}_{p^s}}.\]
Otherwise, again by division algorithm $q_{r-2}(x)=q_{r-1}(x)\hat{b}_{r-1}(x)+p_{r-1}(x)$ with $\mbox{deg}(p_{r-1}(x))<\mbox{deg}(\hat{b}_{r-1}(x))$. Therefore,
\[q_{r-2}(x)\Big(\prod_{t=0}^{r-2}\hat{b}_t(x)\Big)*(l(x)|\mathcal{B}(x)|0)= q_{r-1}(x)\Big(\prod_{t=0}^{r-1}\hat{b}_t(x)\Big)*(l(x)|\mathcal{B}(x)|0)+ p_{r-1}(x)\Big(\prod_{t=0}^{r-2}\hat{b}_t(x)\Big)*(l(x)|\mathcal{B}(x)|0).\]
Clearly, on the one hand, $p_{r-1}(x)\Big(\prod_{t=0}^{r-2}\hat{b}_t(x)\Big)*(l(x)|\mathcal{B}(x)|0)\in \langle B_{r-1}\rangle_{\mathbb{Z}_p^s}$. On the other hand, \[q_{r-1}(x)\Big(\prod_{t=0}^{r-1}\hat{b}_t(x)\Big)*(l(x)|\mathcal{B}(x)|0)= q_{r-1}(x)*\Big(\phi\Big(\frac{x^{\beta}-1}{b_{r-1}(x)}\Big)l(x)|0|0\Big).\]
From Lemma \ref{lemma4.6}, we have $\mathcal A(x)|\phi\Big(\frac{x^{\beta}-1}{b_{r-1}(x)}\Big)l(x)$. Hence, we obtain that \[q_{r-1}(x)\Big(\prod_{t=0}^{r-1}\hat{b}_t(x)\Big)*(l(x)|\mathcal{B}(x)|0)\in \langle\cup_{i=0}^{r-1}A_i\rangle_{\mathbb{Z}_{p^s}}.\]
Thus, $m_2(x)*(l(x)|\mathcal{B}(x)|0)\in \langle \cup_{i=0}^{r-1}A_i\rangle_{\mathbb{Z}_{p^s}}\cup \langle \cup_{j=0}^{r-1}B_j\rangle_{\mathbb{Z}_{p^s}}.$
\par
Now, we only have to show that $m_3(x)*(l_1(x)|l_2(x)|\mathcal{G}(x))\in \langle S\rangle_{\mathbb{Z}_{p^s}}$. If $\mbox{deg}(m_3(x))<\mbox{deg}(\hat{g}_0(x))$, then $m_3(x)*(l_1(x)|l_2(x)|\mathcal{G}(x))\in \langle G_0\rangle_{\mathbb{Z}_{p^s}}$ and $c(x)\in \langle S\rangle_{\mathbb{Z}_{p^s}}$. Otherwise, by division algorithm, $m_3(x)= d_0(x)\hat{g}_0(x)+e_0(x)$ with $\mbox{deg}(e_0(x))<\mbox{deg}(\hat{g}_0(x))$. Then, \[m_3(x)*(l_1(x)|l_2(x)|\mathcal{G}(x))=d_0(x)\hat{g}_0(x)*(l_1(x)|l_2(x)|\mathcal{G}(x))+e_0(x)*(l_1(x)|l_2(x)|\mathcal{G}(x)).\] 
Clearly,  $e_0(x)*(l_1(x)|l_2(x)|\mathcal{G}(x))\in \langle G_0\rangle_{\mathbb{Z}_{p^s}}$. If $\deg (d_0(x))< \deg (\hat{g}_1(x))$, then $d_0(x)\hat{g}_0(x)*(l_1(x)|l_2(x)|$ $ \mathcal{G}(x))\in \langle G_1\rangle_{\mathbb{Z}_{p^s}}$.\par
In the worst-case scenario and by applying similar arguments, one can obtain that \[m_3(x)*(l_1(x)|l_2(x)|\mathcal{G}(x))\in \langle S\rangle_{\mathbb{Z}_{p^s}},\] 
if $d_{s-2}(x)(\prod_{t=0}^{s-2}\hat{g}_t(x))*(l_1(x)|l_2(x)|\mathcal{G}(x))\in \langle S\rangle_{\mathbb{Z}_{p^s}}.$ It is clear, if $\mbox{deg}(d_{s-2}(x))< \mbox{deg}(\hat{g}_{s-1}(x))$ then 
\[q_{s-2}(x)\Big(\prod_{t=0}^{s-2}\hat{g}_t(x)\Big)*(l_1(x)|l_2(x)|\mathcal{G}(x))\in \langle G_{s-1}\rangle_{\mathbb{Z}_{p^s}}\] 
and we are done. Otherwise, again by division algorithm, $d_{s-2}(x)=d_{s-1}(x)\hat{g}_{s-1}(x)+e_{s-1}(x)$ with $\mbox{deg}(e_{s-1}(x)(x))<\mbox{deg}(\hat{g}_{s-1}(x))$. Therefore,
\begin{align*}
  d_{s-2}(x)\Big(\prod_{t=0}^{s-2}\hat{g}_t(x)\Big)*(l_1(x)|l_2(x)|\mathcal{G}(x))&=d_{s-1}(x)(x)\Big(\prod_{t=0}^{s-1}\hat{g}_t(x)\Big)*(l_1(x)|l_2(x)|\mathcal{G}(x))\\
  &~~~+e_{s-1}(x)\Big(\prod_{t=0}^{s-2}\hat{g}_t(x)\Big)*(l_1(x)|l_2(x)|\mathcal{G}(x)).  
\end{align*} 

Clearly, we have $e_{s-2}(x)\Big(\prod_{t=0}^{s-2}\hat{g}_t(x)\Big)*(l_1(x)|l_2(x)|\mathcal{G}(x))\in \langle G_{s-1}\rangle_{\mathbb{Z}_{p^s}}$ and
\begin{align*}
 q_{s-1}(x)\Big(\prod_{t=0}^{s-1}\hat{g}_t(x)\Big)*(l_1(x)|l_2(x)|\mathcal{G}(x))&= q_{s-1}(x)\Big(\prod_{t=0}^{s-1}\hat{g}_t(x)\Big)*(l_1(x)|l_2(x)|0)\\
 &= q_{s-1}(x)*\Big(\phi\Big(\frac{x^{\gamma}-1}{g_{s-1}(x)}\Big)l_1(x)|\phi\Big(\frac{x^{\gamma}-1}{g_{s-1}(x)}\Big)l_2(x)|0\Big).
\end{align*}

Moreover, from Lemma \ref{lemma4.6}, there exists some $Q(x)\in \mathbb{Z}_{p^r}[x]$ such that $Q(x)\mathcal B(x)= \phi\Big(\frac{x^{\gamma}-1}{g_{s-1}(x)}\Big)l_2(x)$ and 
\begin{align*}
  \mathcal{A}(x)|(Q(x)l(x)-\phi\Big(\frac{x^{\gamma}-1}{g_{s-1}(x)}\Big)l_1(x))&\implies \mathcal{A}(x)\lambda(x)= Q(x)l(x)-\phi\Big(\frac{x^{\gamma}-1}{g_{s-1}(x)}\Big)l_1(x)\\
  &\implies \phi\Big(\frac{x^{\gamma}-1}{g_{s-1}(x)}\Big)l_1(x)= Q(x)l(x)- \mathcal{A}(x)\lambda(x).
\end{align*}
Therefore,
\begin{align*}
  q_{s-1}(x)&*\Big(\phi\Big(\frac{x^{\gamma}-1}{g_{s-1}(x)}\Big)l_1(x)|\phi\Big(\frac{x^{\gamma}-1}{g_{s-1}(x)}\Big)l_2(x)|0\Big)\\&= q_{s-1}(x)*(Q(x)l(x)- \mathcal{A}(x)\lambda(x)|Q(x)\mathcal B(x)|0)\\
  &= q_{s-1}(x)[Q(x)(l(x)|\mathcal{B}(x)|0)- \lambda(x)(\mathcal{A}(x)|0|0)].
\end{align*}

Hence, $q_{s-1}(x)(\prod_{t=0}^{s-1}\hat{g}_t(x))*(l_1(x)|l_2(x)|\mathcal{G}(x))=\eta_1(x)*(\mathcal{A}(x)|0|0)+\eta_2(x)*(l(x)|\mathcal{B}(x)|0)\in \linebreak \left\langle \big(\bigcup_{i=0}^{r-1}A_i\big)\cup \big(\bigcup_{j=0}^{r-1}B_j\big)\right\rangle_{\mathbb{Z}_{p^s}}.$ Thus, we conclude that $c(x)\in \langle S\rangle_{\mathbb{Z}_{p^s}}$ and $S$ forms a minimal generating set for $C$. Furthermore, $|C|$ can be calculate from Corollary \ref{cor1}.
 \end{proof}
 \begin{example}\em
 Let $R_{3,3,3}=\mathbb{Z}_2[x]/\langle x^3-1\rangle \times\mathbb{Z}_2[x]/\langle x^3-1\rangle\times\mathbb{Z}_4[x]/\langle x^3-1\rangle$. Consider a $\mathbb{Z}_2\mathbb{Z}_2\mathbb{Z}_4$-additive cyclic code $C$ of block length $(3,3,3)$ generated by
 $$
 \{(1+x+x^2|0|0),(1|x-1|(1+x+x^2)+2)\},
 $$
 where $a_0(x)=1+x+x^2,~b_0(x)=x^3-1,~l_1(x)=1,~l_2(x)=x-1,~g_0(x)=1+x+x^2$ and $g_1(x)=1$. Then, $A_0=\{(1+x+x^2|0|0)\}$,  $G_0=\{(1|x-1|3+x+x^2)\}$ and $G_1=\{(x-1|0|2x-2), (x^2-x|0|2x^2-2x)\}$ forms a minimal generating set for $C$ and $|C|= 2^{1+\sum_{k=0}^{2-1}(s-k)\mbox{deg}(\hat{g}_k(k))}= 2^5=32.$
 \end{example}

 \section{Duality of $\mathbb{Z}_{p^r}\mathbb{Z}_{p^r}\mathbb{Z}_{p^s}$-Additive Cyclic Codes}
In this section, we study the duality of  $\mathbb{Z}_{p^r}\mathbb{Z}_{p^r}\mathbb{Z}_{p^s}$-additive cyclic codes. We discuss some results that determine the relationship between $\mathbb{Z}_{p^r}\mathbb{Z}_{p^r}\mathbb{Z}_{p^s}$-additive cyclic codes and their duals. \par
 Let $C$ be a $\mathbb{Z}_{p^r}\mathbb{Z}_{p^r}\mathbb{Z}_{p^s}$-additive cyclic code and $C^{\perp}$ be its dual. Taking any codeword $\mathbf{u}$ of $C^{\perp}$, $\mathbf{u}\cdot\mathbf{u'}=0$ for all $\mathbf{u'}\in C$. Since $\mathbf{u'}$ is a codeword of $C$, we know that $\mathbf{u'}^{(-1)}$ is also a codeword. So, $\mathbf{u'}^{(-1)}\cdot \mathbf{u}= \mathbf{u'}\cdot \mathbf{u}^{(1)}=0$ for all $\mathbf{u}\in C$. Therefore,  $\mathbf{u}^{(1)}\in C^{\perp}$, hence $C^{\perp}$ is also a $\mathbb{Z}_{p^r}\mathbb{Z}_{p^r}\mathbb{Z}_{p^s}$-additive cyclic code. So we obtain the following result.
 \begin{proposition}\label{pr6.1}
 Let $C$ be a $\mathbb{Z}_{p^r}\mathbb{Z}_{p^r}\mathbb{Z}_{p^s}$-additive cyclic code. Then its dual $C^{\perp}$ is also $\mathbb{Z}_{p^r}\mathbb{Z}_{p^r}\mathbb{Z}_{p^s}$-additive cyclic code.
 \end{proposition}

\begin{proposition}\label{pr6.2}
 Let $C \subseteq \mathbb{Z}^{\alpha}_{p^r}\times \mathbb{Z}^{\beta}_{p^r}\times \mathbb{Z}^{\gamma}_{p^s}$ be a $\mathbb{Z}_{p^r}\mathbb{Z}_{p^r}\mathbb{Z}_{p^s}$-additive cyclic code of block length $(\alpha,\beta,\gamma)$ and $C$ has the generator polynomials given in Theorem \ref{th 4.4}. Then,
 \[|C^{\perp}|= p^{\sum_{i=1}^{r}i\deg(\hat{a}_i(x))+\sum_{j=1}^{r}j\deg(\hat{b}_j(x))+\sum_{k=1}^{s}k\deg(\hat{g}_k(x))}.\]
\end{proposition}
 \begin{proof}
 It is known that $|\mathbb{Z}^{\alpha}_{p^r}\times \mathbb{Z}^{\beta}_{p^r}\times \mathbb{Z}^{\gamma}_{p^s}|= |C||C^\perp|= p^{r\alpha+r\beta+s\gamma}$ and suppose that $|C^{\perp}|= p^t$ for some positive integer $t$. By theorem \ref{th5.1}, we have $|C|=p^{\sum_{i=0}^{r-1}(r-i)\mbox{deg}(\hat{a}_i(x))+\sum_{j=0}^{r-1}(r-j)\mbox{deg}(\hat{b}_j(x))+\sum_{k=0}^{s-1}(s-k)\mbox{deg}(\hat{g}_k(x))}$, this implies
\begin{align*}
    t&= r\alpha+r\beta+s\gamma-\big(\sum_{i=0}^{r-1}(r-i)\mbox{deg}(\hat{a}_i(x))+\sum_{j=0}^{r-1}(r-j)\mbox{deg}(\hat{b}_j(x))+\sum_{k=0}^{s-1}(s-k)\mbox{deg}(\hat{g}_k(x))\big)\\
    &= \sum_{i=1}^{r}i~\mbox{deg}(\hat{a}_i(x))+\sum_{j=1}^{r}j~\mbox{deg}(\hat{b}_j(x))+\sum_{k=1}^{s}k~\mbox{deg}(\hat{g}_k(x)).
\end{align*}
Hence, we prove the required.
 \end{proof}
 By using Proposition \ref{pr6.2}, we obtain the following relation between the $\mathbb{Z}_{p^r}\mathbb{Z}_{p^r}\mathbb{Z}_{p^s}$-additive cyclic code $C$ and the
dual of its dual.
\begin{proposition}\label{pr6.3}
 Let $C$ be a $\mathbb{Z}_{p^r}\mathbb{Z}_{p^r}\mathbb{Z}_{p^s}$-additive cyclic code of block length $(\alpha,\beta,\gamma)$ and $C^\perp$ be its dual. Then $C={(C^\perp)}^\perp$.
 \end{proposition}
 \begin{proof}
 On the one hand, by the definition of dual code, we have $C\subseteq{(C^\perp)}^\perp$.\par
 On the other hand, by Proposition \ref{pr6.1}, we know that ${(C^\perp)}^\perp$ is also a $\mathbb{Z}_{p^r}\mathbb{Z}_{p^r}\mathbb{Z}_{p^s}$-additive cyclic code. Therefore, we have $|C^\perp||{(C^\perp)}^\perp|=p^{r\alpha+r\beta+s\gamma}$. So, by Proposition \ref{pr6.2}, we get $$|{(C^\perp)}^\perp|= p^{\sum_{i=0}^{r-1}(r-i)\mbox{deg}(\hat{a}_i(x))+\sum_{j=0}^{r-1}(r-j)\mbox{deg}(\hat{b}_j(x))+\sum_{k=0}^{s-1}(s-k)\mbox{deg}(\hat{g}_k(x))},$$ which implies that $|C|=|{(C^\perp)}^\perp|$. Hence, we deduce that $C={(C^\perp)}^\perp$.
 \end{proof}
 Let $h(x)=h_0+h_1x+h_2x^2+\cdots+h_nx^n$ be a polynomial in $\mathbb{Z}_{p^s}[x]$, then its $reciprocal~polynomial~ h^*(x)$ defined as $h^*(x)= h_n+h_{n-1}x+\cdots+h_1x^{n-1}+h_0x^n$ i.e. $h^*(x)= x^{\mbox{deg}(h(x))}h(1/x)$. Moreover,  we denote the polynomial $\sum_{i=0}^{m-1}x^i$ by $\theta_m(x)$. From now on, we suppose that $L=\mbox{lcm}[\alpha,\beta,\gamma].$

 \begin{definition}
For any two elements $\mathbf{u}(x)=(u_1(x)|u_2(x)|u_3(x))$ and $\mathbf{w}(x)=(w_1(x)|w_2(x)|w_3(x))$ of $R_{\alpha,\beta,\gamma}$,  we define 
$\bullet: R_{\alpha,\beta,\gamma} \times R_{\alpha,\beta,\gamma}\rightarrow \mathbb{Z}_{p^s}[x]/\langle x^L-1\rangle$ such that
\begin{eqnarray*}
\bullet (\mathbf{u}(x),\mathbf{w}(x))&=& p^{s-r}\epsilon(u_1(x)w^*_1(x))\theta_{\frac{L}{\alpha}}(x^{\alpha})x^{L-1-\deg(w_1(x))} \\
&+& p^{s-r}\epsilon(u_2(x)w^*_2(x))\theta_{\frac{L}{\beta}}(x^{\beta})x^{L-1-\deg(w_2(x))}\\
&+& u_3(x)w^*_3(x)\theta_{\frac{L}{\gamma}}(x^{\gamma})x^{L-1-\deg(w_3(x))}\pmod{(x^L-1)}.
\end{eqnarray*}
\end{definition}
The map $\bullet$ is a linear map in each of its argument. That is, the map $\bullet$ is a bilinear map between $\mathbb{Z}_{p^s}[x]$-modules.
This map can be seen as a generalized map of the given map in \cite[{Definition 5.2}]{224}. Now on wards, we denote $\bullet (\mathbf{u}(x),\mathbf{w}(x))$ by $\mathbf{u}(x)\bullet \mathbf{w}(x)$. Note that $\mathbf{u}(x)\bullet \mathbf{w}(x) \in \mathbb{Z}_{p^s}[x]/\langle x^L-1\rangle$.
\begin{proposition}\label{l6.5}
Let $\mathbf{u}$ and $\mathbf{w}$ be two vectors of $\mathbb{Z}^{\alpha}_{p^r}\times \mathbb{Z}^{\beta}_{p^r}\times \mathbb{Z}^{\gamma}_{p^s}$ with the associated polynomials $\mathbf{u}(x)=(u_1(x)|u_2(x)|u_3(x))$ and $\mathbf{w}(x)=(w_1(x)|w_2(x)|w_3(x))$, respectively. Then $\mathbf{u}$ is orthogonal to $\mathbf{w}$ and all its cyclic shifts if and only if $\mathbf{u}(x)\bullet \mathbf{w}(x)= 0\pmod{(x^L-1)}$.
\end{proposition}
\begin{proof}
Let $\mathbf{u}=(u_{1,0},u_{1,1},\ldots, u_{1,\alpha-1}|u_{2,0},u_{2,1},\ldots,u_{2,\beta-1}|u_{3,0},u_{3,1},\ldots,u_{3,\gamma-1}),~\mathbf{w}=(w_{1,0},w_{1,1}\linebreak,\ldots, w_{1,\alpha-1}|w_{2,0},w_{2,1},\ldots,w_{2,\beta-1}|w_{3,0},w_{3,1},\ldots,w_{3,\gamma-1})$. Consider that $\mathbf{w}^{(i)}=(w_{1,-i},w_{1,1-i},\ldots,\linebreak
w_{1,\alpha-1-i}|w_{2,-i},w_{2,1-i},\ldots,w_{2,\beta-1-i}|w_{3,-i},w_{3,1-i},\ldots,w_{3,\gamma-1-i})$ is the $i$th cyclic shift of vector $\mathbf{w}$, where $i=0,1,2,\ldots, L-1$.  Then $\mathbf{u}\cdot \mathbf{w}^{(i)}=0$ if and only if\[p^{s-r}\sum_{j=0}^{\alpha-1}\epsilon(u_{1,j}w_{1,j-i})+p^{s-r}\sum_{k=0}^{\beta-1}\epsilon(u_{2,k}w_{2,k-i})+\sum_{t=0}^{\gamma-1}u_{3,t}w_{3,t-i}=0\] 
Let $A_i=p^{s-r}\sum_{j=0}^{\alpha-1}\epsilon(u_{1,j}w_{1,j-i})+p^{s-r}\sum_{k=0}^{\beta-1}\epsilon(u_{2,k}w_{2,k-i})+\sum_{t=0}^{\gamma-1}u_{3,t}w_{3,t-i}$. Thus, we have
\begin{eqnarray*}
\mathbf{u}(x)\bullet \mathbf{w}(x)&=& p^{s-r}\theta_{\frac{L}{\alpha}}(x^{\alpha})\sum_{a=0}^{\alpha-1}\sum_{j=0}^{\alpha-1}\epsilon(u_{1,j}w_{1,j-a})x^{L-1-a}+ p^{s-r}\theta_{\frac{L}{\beta}}(x^{\beta})\sum_{b=0}^{\beta-1}\sum_{k=0}^{\beta-1}\epsilon(u_{2,k}w_{2,k-b})x^{L-1-b}\\
 &+&\theta_{\frac{L}{\gamma}}(x^{\gamma})\sum_{c=0}^{\gamma-1}\sum_{t=0}^{\gamma-1}(u_{3,t}w_{3,t-c})x^{L-1-c }\\
&=& \sum_{i=0}^{L-1}A_ix^{L-1-i}\pmod{(x^L-1)}.
\end{eqnarray*}
 Therefore, $\mathbf{u}(x)\bullet \mathbf{w}(x)=0\pmod{(x^L-1)}$ if and only if $A_i=0$, for $i=0,1,\ldots, L-1.$
\end{proof}
Proposition \ref{l6.5} shows that the map $\bullet$  is the corresponding polynomial operation to the inner product of vectors. \par 
\vskip 5pt
\noindent
As an application of our study, we consider $p=2$ and $r=s=1$, and using the generator polynomials studied in Theorem \ref{th 4.4}, we obtain several optimal binary codes according to \cite{Gra} in Table \ref{t1}. Note that if $p=2$ and $r=s=1$, then we get that $C$ forms a $\mathbb{Z}_2\mathbb Z_2\mathbb{Z}_2$-additive cyclic codes.
 
\begin{table}
  \begin{center}
    \caption{Some optimal binary codes. \label{t1}}
    \renewcommand{\arraystretch}{1.3}
     \vspace{0.5cm}
\begin{tabular}{|c|c|c|}
\hline
$(\alpha,\beta,\gamma)$& Generator polynomials of $C$ & Optimal binary codes \\
\hline
$(3,1,3)$ &  \parbox{10cm}{$\mathcal A(x)=x^2+x+1,~\mathcal B(x)=x+1,~l(x)=0,\\ l_1(x)=x+1,~ l_2(x)=1,~\mathcal G(x)=1$.} & $[7,4,3]$ \\
 \hline
$(1,1,7)$ &  \parbox{10cm}{$\mathcal A(x)=x+1,~\mathcal B(x)=x+1,~l(x)=0,\\ l_1(x)=1,~ l_2(x)=1,~\mathcal G(x)=x^3+x^2+1$.} & $[9,4,4]$ \\
 \hline
 $(7,1,7)$ &  \parbox{10cm}{$\mathcal A(x)=x^7+1,~\mathcal B(x)=x+1,~l(x)=0,\\ l_1(x)=x^4+x^3+x^2+1,~ l_2(x)=0,~\mathcal G(x)=x^4+x^3+x^2+1$.} & $[15,3,8]$ \\
 \hline
 $(1,1,15)$ &  \parbox{10cm}{$\mathcal A(x)=x+1,~\mathcal B(x)=x+1,~l(x)=0,\\ l_1(x)=1,~ l_2(x)=1,~\mathcal G(x)=x^8+x^7+x^6+x^4+1$.} & $[17,7,6]$ \\
 \hline
 $(1,1,15)$ &  \parbox{10cm}{$\mathcal A(x)=x+1,~\mathcal B(x)=x+1,~l(x)=0,\\ l_1(x)=1,~ l_2(x)=1,~\mathcal G(x)=x^4+x+1$.} & $[17,11,4]$ \\
 \hline
 $(1,1,17)$ &  \parbox{10cm}{$\mathcal A(x)=x+1,~\mathcal B(x)=x+1,~l(x)=0,\\ l_1(x)=1,~ l_2(x)=1,~\mathcal G(x)=x^8+x^5+x^4+x^3+1$.} & $[19,9,6]$ \\
 \hline

 $(7,7,7)$ &  \parbox{10cm}{$\mathcal A(x)=x^7+1,~\mathcal B(x)=x^7+1,~l(x)=0,\\ l_1(x)=x^4+x^3+x^2+1,~ l_2(x)=x^4+x^3+x^2+1,~\mathcal G(x)=x^4+x^3+x^2+1$.} & $[21,3,12]$ \\
 \hline

\end{tabular}
  \end{center}
\end{table}

 \section{Conclusion}
In this paper, we have discussed the algebraic structure of $\mathbb Z_{p^r}\mathbb Z_{p^r}\mathbb Z_{p^s}$-additive cyclic codes, where $r\leq s$. These codes can be viewed as $\mathbb{Z}_{p^s}[x]$-submodules of the module $\mathbb{Z}_{p^r}[x]/\langle x^{\alpha}-1\rangle \times \mathbb{Z}_{p^r}[x]/\langle x^{\beta}-1\rangle\times \mathbb{Z}_{p^s}[x]/\langle x^{\gamma}-1\rangle$.  We study the generator polynomials of $\mathbb{Z}_{p^r}\mathbb{Z}_{p^r}\mathbb{Z}_{p^s}$-additive cyclic codes in Theorem \ref{th 4.4} and some necessary and sufficient conditions on the generator polynomials in next few results. In Example \ref{ex:Z2Z2Z4} and Remark \ref{R4.12}, we explain that the generator polynomials are given in \cite{mosta, Gao2016, 224} do not generate a general additive cyclic code.  We also discuss the structure of separable $\mathbb{Z}_{p^r}\mathbb{Z}_{p^r}\mathbb{Z}_{p^s}$-additive cyclic codes and determine their generator polynomials also. In Theorem \ref{th4.12}, we classify some cases of generator polynomials of this family of codes.  After that, We determine the minimal generating sets of $\mathbb{Z}_{p^r}\mathbb{Z}_{p^r}\mathbb{Z}_{p^s}$-additive cyclic codes and their size. In the last section, we discuss the duality of $\mathbb{Z}_{p^r}\mathbb{Z}_{p^r}\mathbb{Z}_{p^s}$-additive cyclic codes and obtain the relation between the $\mathbb{Z}_{p^r}\mathbb{Z}_{p^r}\mathbb{Z}_{p^s}$-additive cyclic codes and the dual of their duals. By using the generator polynomials of this family of codes, some optimal binary codes are obtained in Table \ref{t1}.

\section*{Acknowledgement}
First author is thankful to Spanish MINECO under Grant PID2019-104664GB-I00 (AEI/FEDER, UE) for partially support.
%\section*{Conflict of interest}
%On behalf of all authors, the corresponding author states that there is no conflict of interest.

 \end{document}